\documentclass[conference,review]{IEEEtran}

\IEEEoverridecommandlockouts
\usepackage{cite}
\usepackage{amsmath,amssymb,amsfonts}
\usepackage{mathrsfs}

\usepackage{algorithmic}
\usepackage{graphicx}
\usepackage{textcomp}
\usepackage{xcolor}
\def\BibTeX{{\rm B\kern-.05em{\sc i\kern-.025em b}\kern-.08em
    T\kern-.1667em\lower.7ex\hbox{E}\kern-.125emX}}

\usepackage{tikz-cd}
\usepackage{stmaryrd}
\usepackage{hyperref}
\usepackage[all]{xy}

\newcommand{\dua}{\mathord{\hbox{\makebox[0pt][l]{\raise .6mm
                           \hbox{$\uparrow$}}$\uparrow$}}}

\newcommand{\remove}[1]{}

\newcommand{\bfomega}{{\boldsymbol{\omega}}}
\newcommand{\lo}[1]{\operatorname{#1} \,}
\newcommand{\sample}{\lo{sample}}
\newcommand{\score}{\lo{score}}
\newcommand{\inter}[1]{\lo{c_{#1}}}
\newcommand{\ldot}{\, . \,}
\newcommand{\rmin}{\lo{min}}

\newcommand{\cplus}{\lo{(+)}}

\newcommand{\rmax}{\lo{max}}

\newcommand{\ttt}{\lo{tt}}
\newcommand{\ff}{\lo{f\!f}}

\newcommand{\pr}{\lo{pr}}
\newcommand{\integral}{\lo{int}}

\newcommand{\Comp}{\mathrm{Comp}}

\newcommand{\nat}{\mathbb{N}}
\newcommand{\reale}{\overline{\realLine}^+}
\newcommand{\interval}{\mathbb{I}}
\newcommand{\realLine}{\mathbb{R}}
\newcommand{\realDom}{\mathbb{IR}}
\newcommand{\Dom}{\mathbb{D}}
\newcommand{\BC}{\textrm{\bf{BC}}}
\newcommand{\D}{\textrm{\bf{D}}}

\newcommand{\PD}{\textrm{\bf{PER}}}

\newcommand{\id}{{\sf id}}

\newcommand{\sle}{\sqsubseteq}

\newcommand{\sem}[1]{\llbracket{#1} \rrbracket}
\newcommand{\bsem}[1]{\mathcal{B}\sem{#1}}
\newcommand{\esem}[1]{\mathcal{E}\sem{#1}}
\newcommand{\lifted}[1]{(#1)_\bot}

\usepackage{amsthm}
\newtheorem{theorem}{Theorem}[section]
\newtheorem{lemma}[theorem]{Lemma}
\newtheorem{corollary}[theorem]{Corollary}

\newtheorem{example}[theorem]{Example}
\newtheorem{definition}[theorem]{Definition}
\newtheorem{proposition}[theorem]{Proposition}

\begin{document}

\title{A Domain-Theoretic Framework for Conditional Probability and Bayesian Updating in Programming
}

\author{\IEEEauthorblockN{Pietro Di Gianantonio}
\IEEEauthorblockA{
\textit{Univesity of Udine}\\
Udine, Italy \\
pietro.digianantonio@uniud.it}
\and
\IEEEauthorblockN{Abbas Edalat}
\IEEEauthorblockA{
\textit{Imperial College London}\\
London, UK \\
ae@ic.ac.uk}
}

\maketitle

\begin{abstract}
We present a domain-theoretic framework for probabilistic programming that provides a constructive definition of conditional probability and addresses  computability challenges previously identified in the literature. We introduce a novel approach based on an observable notion of events that enables computability. We examine two methods for computing conditional probabilities—-one using conditional density functions and another using trace sampling with rejection—-and prove they yield consistent results within our framework. We implement these ideas in a simple probabilistic functional language with primitives for sampling and evaluation, providing both operational and denotational semantics and proving their consistency. Our work provides a rigorous foundation for implementing conditional probability in probabilistic programming languages.
\end{abstract}
\section{Introduction}

Probabilistic programming languages (PPLs) have emerged as powerful tools for modelling complex systems involving uncertainty, enabling programmers to express probabilistic models and perform statistical inference using familiar programming constructs. Notable examples include BLOG \cite{Milch2005}, Church \cite{Goodman2008}, Anglican \cite{Wood2014}, and Stan \cite{Carpenter2017}, with applications ranging from artificial intelligence to scientific modelling.

Functional programming features have proven valuable in this context, offering natural abstractions for expressing probabilistic models and compositional reasoning about probabilistic computations. However, providing formal semantics for such languages presents significant challenges, particularly when dealing with continuous probability distributions.

Traditional approaches to denotational semantics for probabilistic languages have encountered the well-known Jung-Tix problem \cite{JungT98}: the difficulty of constructing a Cartesian closed category with probabilistic domain constructors. This has led researchers to explore various alternatives to classical domain theory, such as quasi-Borel spaces \cite{HeunenKSY17} and measurable spaces with random elements \cite{Staton2017}.

An alternative approach is to model probabilistic computations in terms of random variables rather than probability distributions. 
Random variables have been employed in denotational semantics of programming languages in various work \cite{BacciFKMPS18,GoubaultV11,HeunenKSY17,HuotLMS23,Mislove16,Scott14,Vakar2019}. However, defining an appropriate monad structure for random variables has proven challenging. Recently, \cite{DiGianantonio2024} proposed a solution based on domains with partial equivalence relations (PER), a concept used previously in several contexts including in~\cite{AP90,bauer2004equilogical,AL00}.

Another challenge in probabilistic computation is the treatment of conditional probability. The seminal result of Ackerman et al.\cite{ackerman2019} demonstrates that conditional probability is not computable in its classical formulation, even for simple cases involving continuous distributions. This non-computability arises from the fact that conditioning on a Borel subset that is not an open set, though well-defined in classical probability theory, cannot be effectively computed: its membership predicate is not an observable property or semi-decidable predicate in the sense of observational logic~\cite{smyth1983power,abramsky1991domain,vickers1989topology}. This theoretical limitation has profound implications for the design of probabilistic programming languages and their semantics.

In this paper, we extend the framework in~\cite{DiGianantonio2024} in several directions, with a primary focus on addressing the challenges posed by Ackerman et al.

We address these challenges by introducing a constructive notion of events--in which an event and its opposite are disjoint pairs of open sets and thus are observable properties--leading to a computable treatment of conditional probability that avoids the pathological cases identified by Ackerman et al. The domain of disjoint pairs of open sets of a topological space, first used in~\cite{edalat2002foundation} to develop a computable framework for solid modelling and computational geometry, establishes a proper domain of events, which makes conditional probability computable. Furthermore, we show how Bayesian updating can be implemented using weighted sampling through a score operator, necessitating a shift from probability distributions to more general measures.

A distinguishing feature of our work is the comprehensive treatment of continuous distributions combined with constructive real number computation. While previous approaches have typically relied on floating-point approximations, we develop a rigorous framework for exact real number computation in the context of probabilistic programming.
Specifically, we present a functional programming language that integrates three key features: continuous distributions, score-based weighted sampling, and exact real number computation. We provide both denotational and operational semantics, introducing new techniques to handle the interaction between these features. The correspondence between these semantic approaches is established through an adequacy theorem.

The paper is organized as follows: Section II presents the domain theory and measure theory basics. Section III introduces our valuation monad for PER domains. Section IV develops the treatment of conditional probability, addressing the computability issues highlighted by Ackerman et al. Section V presents extended random variables. Section VI deals with Bayesian updating. Section VII describes our probabilistic functional language and its semantics.


\subsection{Related work}
There exists a large literature on denotational models for probabilistic computation. Many of these works differ substantially from ours in two aspects. First, they do not use random variables and, instead, model probabilistic computation using probability measures~\cite{DanosE11,EhrhardPT18,HeunenKSY17,Vakar2019}, or continuous distributions~\cite{GoubaultJT23,JiaLMZ21}; second, these works do not use Scott domains but instead employ a larger class of dcpo's~\cite{JiaLMZ21,GoubaultJT23}, probabilistic coherence spaces~\cite{EhrhardPT18}, or categories built from measurable spaces~\cite{HeunenKSY17,HuotLMS23,Vakar2019}.  

In the existing literature, there are a limited number of constructions of a probabilistic monad based on random variables. A monad construction quite similar to ours, but not incorporating soft-conditioning, can be found in~\cite{DiGianantonio2024}. In~\cite{Barker16}, random variables are defined as pairs consisting of a domain and a function, but without discussion of monad commutativity. A strong monad of random variables is defined in~\cite{GoubaultV11}, although the construction differs significantly from ours. Domains of random variables with structures similar to our approach are defined in~\cite{Mislove16,Scott14}, but without presenting monad constructions. Random variables are also used to define quasi-Borel spaces by Heunen et al.~\cite{HeunenKSY17}.

The treatment of conditional probability in probabilistic programming has received significant attention. Various approaches to soft conditioning can be found in the literature: Staton~\cite{Staton2017} employs metric spaces, Park et al.~\cite{park2008probabilistic} use sampling functions, and Goubault-Larrecq et al.~\cite{GoubaultJT23} develop a domain-theoretic framework. The relationship between scoring and conditioning has been explored in several works~\cite{Staton2017,Vakar2019}, with different approaches for implementing conditioning in higher-order languages.

\section{Domain-theory and measure-theory basics}
We use standard notions in domain theory, as described in~\cite{abramsky1995domain} and~\cite{gierz2003continuous}. 
A {\em directed complete partial order} (dcpo) $D$ is a partial order in which every (non-empty) directed set $A\subseteq D$ has a lub (least upper bound) or supremum $\sup A$. The {\em way-below relation} $\ll$ in a dcpo $(D,\sqsubseteq)$ is defined by $x\ll y$ if whenever there is a directed subset $A\subseteq D$ with $y\sqsubseteq \sup A$, then there exists $a\in A$ with $x\sqsubseteq a$.  A subset $B\subseteq D$ is a {\em basis} if for all $y\in D$ the set $\{x\in B:x\ll y\}$ is directed with lub $y$. By a {\em domain}, we mean a non-empty dcpo with a countable basis. Domains are also called {\em $\omega$-continuous dcpo's}. In a domain $D$ with basis $B$, we have the interpolation property: the relation $x\ll y$, for $x,y\in D$, implies there exists $z\in B$ with $x\ll z\ll y$. For any $x\in D$, we write $\dua x=\{y\in D: x\ll y\}$, and for any finite subset $A\subset D$, we write $\uparrow\hspace{-.7ex} A=\{x\in D: \exists a\in A.\,a\sqsubseteq x\}$. 

A subset $A\subseteq D$ is {\em bounded} if there exists $d\in D$ such that for all $x\in A$ we have $x\sqsubseteq d$. If any bounded subset of $D$ has a lub then $D$ is called {\em bounded complete}. Thus, a bounded complete domain, also called a {\em Scott domain}, has a bottom element $\bot$, which is the lub of the empty subset. 

We denote the interior and the complement of a subset $S$ of a topological space by $S^\circ$ and ${S}^c$ respectively. The lattice of open sets of a topological space $X$ is denoted by $\mathcal{O}(X)$ For a map $f:X\to Y$, denote the image of any subset $S\subseteq X$ by $f[S]$.

The set of non-empty compact intervals of the real line ordered by reverse inclusion and augmented with the whole real line as bottom is the prototype bounded complete domain for real numbers denoted by $\realDom$, in which $I\ll J$ iff $J \subseteq I^\circ$. It has a basis consisting of all intervals with rational endpoints.  If $I\subseteq \realLine$ is a non-empty real interval, its left and right endpoints are denoted by $I^-$ and $I^+$ respectively; thus if $I$ is compact, $I=[I^-,I^+]$. A {\em dyadic} number is of the form $m/2^n$, for some $m,n\in \nat$.

The Scott topology on a domain $D$ with basis $B$ has sub-basic open sets of the form $\dua b:=\{x\in D:b\ll x\}$ for any $b\in B$. We let $\D$ denote the category of domians with Scott continuous maps, whereas $\BC$ denotes the full sub-category of bounded complete domains. 

If $X$ is any topological space with some open set $O\subseteq X$ and $d\in D$, then the single-step function $d\chi_O:X\to D$, defined by $d\chi_O(x)=d$ if $x\in O$ and $\bot$ otherwise, is a Scott continuous function. The partial order on $D$ induces, by pointwise extension, a partial order on continuous functions of type $X\to D$ with $f\sqsubseteq g$ if $f(x)\sqsubseteq g(x)$ for all $x\in X$.  For any two bounded complete domains $D$ and $E$, the function space $(D\to E)$ consisting of Scott continuous functions from $D$ to $E$ with the extensional order is a bounded complete domain with a basis consisting of lubs of bounded and finite families of single-step functions. 

If $X$ is a topological space such that its lattice $\mathcal{O}(X)$ of open sets is continuous and $D$ is a bounded complete domain, with $f:X\to D$ a Scott continuous and $d\chi_O:X\to D$ a single-step function, then we have~ \cite[Proposition II-4.20(iv)]{gierz2003continuous}
\begin{equation}\label{way-below}d\chi_O\ll f \iff O\ll_{\mathcal{O}(X)} f^{-1}(\dua d)\end{equation}

\subsection{Computability in domain theory} We adopt the notion of effectively given domains as in~\cite{DiGianantonio2024}, which for Scott domains is equivalent to~\cite{plotkin1981post}. We say a domain $D$ is {\em effectively given} with respect to an enumeration $b:{\mathbb N}\to B$ of the countable basis $B\subseteq D$ if the relation $ b_m\ll b_n$ is decidable. If $D$ has a least element $\bot$, we assume $b_0=\bot$. If $D$ is, additionally, bounded complete, we further assume that the relation $b_m\sqcup b_n=b_p$ is decidable. We say $x\in D$ is {\em computable} if the set $\{n:b_n\ll x\}$ is r.e. The latter is equivalent to the existence of a total recursive function $g:\nat\to \nat$ such that $(b_{g(n)})_{n\in \nat}$ is an increasing sequence with $x=\sup_{n\in\nat}b_{g(n)}$. \remove{A sequence $(x_n)_{n\in \nat}$ is said to be {\em computable} if there exists a total recursive function $h$ such that $x_n=x_{h(n)}$ for $n\in \nat$.} A continuous map $f:D\to E$ of effectively given domains $D$, with basis $\{a_0,a_1,\ldots\}$, and $E$, with basis $\{b_0,b_1\ldots\}$, is {\em computable} if the relation $b_n\ll f(a_m)$ is r.e. 

\subsection{Elements of measure theory}
 Recall, say from~\cite{athreya2006measure}, that a measurable space on a set $X$ is given by a $\sigma$-algebra $S_X$ subsets of $X$, i.e., a non-empty family of subsets of $X$ closed under the operations of taking countable unions, countable intersections and complementation. Elements of $S_X$ are called measurable sets. For a topological space $X$, the collection of all Borel sets on X forms a $\sigma$-algebra, known as the $\sigma$-Borel algebra: it is the smallest $\sigma$-algebra containing all open sets (or, equivalently, all closed sets). A map $f: (X,S_X)\to (Y,S_Y)$ of two measurable spaces is measurable, if $f^{-1}(C) \in S_X$ for any $C\in S_Y$. Any continuous function of topological spaces is measurable with respect to the Borel algebras of the two spaces. 

 Let $\overline{\realLine^+}$ be the set of extended non-negative real numbers (including $+\infty$)  with the usual ordering, equipped with is Scott topology, i.e., for any $a\in \realLine^+$, the set $\{x:a<x\}$ is open. The extended set of non-negative rational numbers $\overline{\mathbb{Q}^+}$ is similarly defined. 

A measure on a measure space $(X,S_X)$ is a map $\nu:S_X\to \reale$ with $\nu(\emptyset)=0$ such that $\nu(\bigcup_{i\in \nat}\nu(E_i))=\sum_{i\in \nat}\nu(E_i)$ for any countable disjoint collection of measurable sets $E_i$ for $i\in \nat$. if $f:(X,S_X)\to (Y,S_Y)$ is a measurable map then it induces a measure $\nu\circ r^{-1}$ on $Y$ called the {\em push-forward measure}.

A {\em probability space} $(\Omega,S_\Omega,\nu)$ is given by a probability measure $\nu$ on a measure space $(\Omega,S_\Omega)$, where $\Omega$ is called the {\em sample space}, $\nu(\Omega)=1$, and measurable sets are called {\em events}. A subset $A\subseteq \Omega$ is a {\em null set} if $A\subseteq B\in S_\Omega$ with $\nu(B)=0$. 

Given the probability space $(\Omega,S_\Omega,\nu)$, a random variable on a measurable space $(Y,S_Y)$ is a measurable map $r: (\Omega,S_\Omega,\nu)\to (Y,S_Y)$. The probability of $r\in C$ for $C\in S_Y$ is given by $\operatorname{Pr}(r\in C)= \nu(r^{-1}(C))$.  Two random variables $r_1$ and $r_2$ are {\em independent} if $\operatorname{Pr}(r_1\in C_1,r_2\in C_2)= \operatorname{Pr}(r_1\in C_1)\operatorname{Pr}(r_2\in C_2)$ for all $C_1,C_2\in S_Y$.

Recall that for any measure $\nu$ on a measurable space $X$, any measurable map $w:X\to \realLine^+$ defines a measure, which we denote by $\nu_w$, with $\nu_w(S)=\int_S w\,d\nu$, for any measurable $S\subset X$. This construction has a converse. Recall that a measure $\nu$ on $X$ is $\sigma$-finite if $X$ is the countable union of measrable sets with finite $\nu$-measure.  Now, suppose $\nu_1$ and $\nu_2$ are $\sigma$-finite measures on $X$ such that $\nu_2$ is absolutely continuous with respect to $\nu_1$, i.e., for any measurable set $S\subset X$, we have: $\nu_1(S)=0$ implies $\nu_2(S)=0$. Then there exists, essentially, a unique measurable function $f:X\to [0,\infty)$, called the Radon-Nikodym derivative of $\nu_2$ with respect to $\nu_1$ and denoted by  $f=\frac{d\nu_2}{d\nu_1}$ such that  $\nu_2(S)=\int_S f\,d\nu_1$. In addition, for any measurable map $g:X\to \realLine$ and measurable set $S\subset X$, we have~\cite{rudin1987real,billingsley2017probability}: 
\begin{equation}\label{rel-density-int}
    \int_S g\,d\nu_2=\int_S g\frac{d\nu_2}{d\nu_1}\,d\nu_1=\int_S g\cdot f\,d\nu_1.\end{equation}

Consider the domain $D\in \Dom$, with countable basis $B_D$,  equipped with the sigma algebra induced by the Scott topology as a measure space.  Then $\mathcal{O}(D)$  has countable basis $B_{\mathcal{O}}$ consisting of opens sets $\dua b$ with $b\in B_D$.  A {\em valuation} is a map $\alpha: \mathcal{O}(D)\to \overline{\realLine^+}$ satisfying (i) $\alpha(\emptyset)=0$, (ii) $\alpha(O_1)\leq \alpha(O_2)$ if $O_1\subset O_2$, (iii) $\alpha(O_1)+\alpha(O_2)=\alpha(O_1\cup O_2)+\alpha(O_1\cap O_2)$. A {\em continuous} valuation is a valuation that is continuous with respect to the Scott topology.

We say a continuous valuation on $D\in \D$ is {\em $\sigma$-finite} if there are open sets $O_i$ for $i\in \nat$ with $D\setminus \{\bot\}=\bigcup_{i\in \nat} O_i$ such that $\nu(O_i)<\infty$ for all $i\in \nat$. We know that every bounded continuous valuation on $D$ extends uniquely to a Borel measure on $D$ and that every $\sigma$-finite continuous valuation on $D$ uniquely corresponds to a measure on $D$ except for the arbitrary choice of the measure on the single element $\{\bot\}$~\cite{alvarez2000extension}. Moreover, any $\sigma$-finite continuous valuation on a locally compact Hausdorff space
can be extended to a Borel measure~\cite[Theorem 5.3]{KL05}.

The notion of $\sigma$-finite measure and $\sigma$-finite valuation are related via spaces of maximal points, i.e., a domain $D$ whose set of maximal elements $\operatorname{Max}(D)$ equipped with its relative Scott topology is a Polish space (i.e., a separable, complete and metrizable space)~\cite{lawson1997spaces}. On these domains, $\operatorname{Max}(D)$ is a $G_\delta$ subset, i.e., a countable intersection of open sets.  If a continuous valuation $\nu$ on $D$ has its support on $\operatorname{Max}(D)$~\cite{edalat1995dynamical,lawson1998computation}, i.e., $\nu(D\setminus\operatorname{Max}(D))=0 $, then $\nu$ is $\sigma$-finite iff its restriction $\nu\restriction_{\operatorname{Max}(D)}$ on the maximal elements is a $\sigma$-finite measure.

We let $\mathcal{V}(D)$ denote the set of continuous  valuations on $D\in\D$ with pointwise ordering induced by their values on opens of $D$.  For any $d\in D$, the {\em point valuation} $\delta(d)$ at $d$ is given by $\delta(d)(O)=1$ if $d\in O$ and $0$ otherwise. A {\em simple valuation} $\alpha$ on $D$ is the finite sum of points valuations. Let  $\alpha=\sum_{i\in I}^nr_i\delta(d_i)$ be a simple valuation and $\beta\in \mathcal{V}(D)$.

\begin{proposition}\label{kirch}~\cite{kirch1993bereiche}
    \begin{enumerate}
        \item[(i)] $\alpha\ll \beta$ iff for all $J\subset I$,
        \[\sum_{i\in J}r_i<\beta\left(\bigcup_{i\in J}\dua d_i\right)\]
        \item[(ii)] $\mathcal{V}(D)$ is a domain with a countable basis $B_{\mathcal{V}(D)}$ of simple valuations $\sum_{i\in I}^nr_i\delta(d_i)$, obtained by choosing $r_i\in \mathbb{Q}$ and $d_i\in B_D$, for $i\in \nat$.  
    \end{enumerate}
\end{proposition}

It follows from Propositon~\ref{kirch}(i) that if $D$ is effectively given, then the way-below relation on $\mathcal{V}(D)$ is decidable and thus $\mathcal{V}(D)$ inherits an effective structure. 

Let $B_{ (D\to \overline{\realLine^+})}$ be the countable basis of the function space $ (D\to \overline{\realLine^+}) $ consisting of step functions of the form $g=\sup_{i\in I} r_i\chi_{O_i}$ where $r_i\in {\mathbb{Q}^+}$ and $O_i\in B_{\mathcal{O}(D)}$. Let $\alpha=\sum_{i\in I}^nr_i\delta(d_i)$ be a simple valuation. We now define:
\[F: B_{\mathcal{V}(D)}\times B_{(D\to \overline{\realLine^+}) }\to \mathcal{V}(D)\]
\begin{equation}\label{int-on-basis}(\alpha,g)\mapsto\sum_{d_i\in O_j}r_iq_j\delta(d_i)\end{equation}
We note that $F(\alpha,g)=\lambda O.\int_Og\,d\alpha$. 
It is straightforward to check that, for $\alpha\in B_{\mathcal{V}(D)}$ and $g\in B_{(D\to \overline{\realLine^+})} $, we have:
\begin{equation}\label{cons-ex}
    F(\alpha,g)=\sup \{F(\alpha',g'):\alpha'\ll \alpha,g'\ll g\}.\end{equation}
It follows, from~\cite[Section 2.2.6]{abramsky1995domain},  that $F$ has unique conservative extension to 
\[F:  {\mathcal{V}(D)}\times {(D\to \overline{\realLine^+}) }\to \mathcal{V}(D),\]
given by Equation~\ref{cons-ex} for any $\alpha\in \mathcal{V}(D)$ and $g\in {(D\to \overline{\realLine^+})}$.

\begin{proposition}\label{functionalF}
\begin{itemize}
    \item[(i)] 
The functional $F$ is Scott continuous, and if $\alpha$ is a $\sigma$-finite valuation then $F(\alpha,g)=\lambda O.\int_O g\,d\alpha$.
\item[(ii)] $F$ is computable with respect to any effective structure on $D$, and the resulting effective structures on $\mathcal{V}(D)$ and $(D\to \overline{\realLine^+})$.  
\end{itemize}
\end{proposition}
{
    \begin{proof}
   (i) The Scott continuity of $F$ follows from the construction~\cite{abramsky1995domain}. If $\alpha$ is a $\sigma$-finite valuation then it has a unique extension to a Borel measure~\cite{alvarez2000extension}. Let $\alpha\in \mathcal{V}(D)$ and $g\in (D\to \overline{\realLine^+})$. Assume $\alpha_i\ll \alpha$ for $i\in \nat$ is an increasing sequence of simple valuations with $\alpha=\sup_{i\in \nat}\alpha_i$.  Then, $F(\alpha_i,g)=\lambda O. \int_Og\,d\alpha_i$. By the Choquet's formulation of the Lebesgue integral~\cite{mesiar2010choquet}, we obtain: $\lambda O. \int_Og\,d\alpha_i=\lambda O. \int_0^\infty \alpha_i(g^{-1}((a,\infty])\cap O)\,da$. Since the sequence of maps $\lambda a.\alpha_i(g^{-1}((a,\infty])\cap O)$ for $i\in \nat$ is increasing with $\sup_{i\in\nat} \lambda a.\alpha_i(g^{-1}((a,\infty])\cap O)=\lambda a. \alpha(g^{-1}((a,\infty])\cap O)$, by the monotone convergence theorem we obtain: $F(\alpha,g)=\sup_{i\in \nat} F(\alpha_i,g)=\lambda O. \sup_{i\in \nat}\int_Og\,d\alpha_i=\lambda O. \int_Og\,d\alpha $.

   (ii) Let $\alpha$ be the simple valuation and $g$ the step function in Equation~\ref{int-on-basis}. By Proposition~\ref{kirch}(i)  the relation,
   \[\alpha'\ll F(\alpha,g)=\sum_{d_i\in O_j}r_iq_j\delta(d_i)\]
   is decidable for any simple valuation $\alpha'$. Hence, $F$ is computable. 
\end{proof}
}

\section{Random variable \& valuation monads on {\bf PER}}
We here first recall the definition of the category of PER domains from~\cite{DiGianantonio2024}; then define a parametrised family of random variable monads on this categroy and finally extend the monad introduced in~\cite{kirch1993bereiche} to the catgeory of PER domains.

A {\em PER domain} $(D,\sim)$ is a bounded complete, i.e., Scott domain $D$ with a symmetric and transitive {\em logical} relation $\sim$, i.e.,  satisfying:\\
(i) $\bot\sim \bot$ and (ii)  for increasing chains $(d_i)_{i\in \nat}$ and $(d_i')_{i\in \nat}$,\\ $\forall i\in \nat$ $d_i\sim d'_i\Longrightarrow \sup_{i\in \nat} d_i\sim \sup_{i\in \nat} d'_i$. 

For Scott continuous functions $f,g: D \to E$, we define $f\sim g$ if for all $d,d' \in D$, $d\sim d'\Longrightarrow f(d)\sim f(d')$ and $f(d)\sim g(d)$.
A morphism $[f]:(D,\sim_D)\to (E,\sim_E)$ in the category of PER domains is a partial equivalence class of Scott continuous functions.
The product of PER domains $(D,\sim_D)\times (E,\sim_E)$ is the product of the underlying domains equipped with the PER defined by $(d,e)\sim (d',e')$ if and only if $d\sim_D d'$ and $e\sim_E e'$.
Finally, the function space $(D,\sim_D)\to (E,\sim_E)$ consists of the Scott continuous functions between the underlying domains with the PER defined as above.
The category {\bf PER} of PER domains forms a Cartesian closed category (CCC).
Any domain $D$ is a PER domain by taking identity as a partial equivalence relation. The mapping $D \mapsto (D, =)$ defines a full and faithful functor from the category of domains to the category of PER domains.

An element $d\in (D,\sim)$ is {\em self-related} if $d\sim d$. A Scott open subset $O\subseteq D$ in a PER domain $\langle D,\sim_D\rangle$ is {\em R-open} if it is closed under $\sim_D$, i.e., $d_1 \in O$ and $d_1\sim d_2$ implies $d_2\in O$. The R-open sets  form a topology called R-topology denoted by $\Upsilon_D$.

For the sample space $\Omega$, we take a countably based, locally compact Hausdorff topological space $\Omega$, for which $\mathcal{O}(\Omega)$ is equipped with an effective structure,  with a  computable continuous valuation $\nu:\mathcal{O}(\Omega)\to \reale$ on it. If $\nu$ is $\sigma$-finite then it always extends to a Borel measure on $\Omega$~\cite[Theorem 5.3]{KL05}.

The random variable functor $\mathcal{R}$ is defined on ${\bf PER}$ category by:
{ On objects:} $\mathcal{R}(D,\sim_D)=((\Omega\to D),\sim_{\Omega\to D})$ with \\
$r\sim_{{\Omega\to D}}s$ if (i) $\forall \bfomega\in \Omega.\,r(\bfomega)\sim_D r(\bfomega)\,\&\,s(\bfomega)\sim_D s(\bfomega)$, and (ii) $\forall O\in \Upsilon_D.\,\nu(r^{-1}(O))=\nu(s^{-1}(O))$. 
And {on morphisms:} $[f]:(D,\sim_D)\to (E,\sim_E)$ with the map 
$\mathcal{R}[f]:((\Omega\to D),\sim_{\mathcal{R}(D,\sim)})\to ((\Omega\to E),\sim_{\mathcal{R}(E,\sim)})$ as
$\mathcal{R}[f]=[\lambda r.\,f\circ r]$.

  Let $h:(\Omega,\nu)\to (\Omega^2,\nu\times \nu)$ be a continuous and measure-preserving map that is a bijection on a set of full $\nu$-measure on $\Omega$. It is shown in ~\cite{DiGianantonio2024} that such $h$ makes ${\mathcal R}$ into a monad on the category {\bf PER}:
 $\mathcal{R}: {\bf PER}\to {\bf PER}$ with natural transformations 
{\bf Unit:} $\eta_D: D\to (\Omega\to D)$ with $\eta_D(d)(\bfomega)=d$.
 {\bf Flattening:} $\mu_D:(\Omega\to (\Omega\to D))\to (\Omega\to D)$ with 
$\mu_D(r)(\bfomega)=r(h_1(\bfomega))(h_2(\bfomega))$. 
  
In~\cite{DiGianantonio2024}, two standard probability spaces and monads were considered for $(\Omega,\nu)$ as follows: (i) Take $\nu$ to be the probability measure on $\Omega=\{0,1\}^\nat$ with the measure of any cylinder set of length $n$ to be $1/2^n$.
 For $\bfomega\in \{0,1\}^\nat$, define 
  $h(\bfomega) = (\bfomega^{\operatorname{e}}, \bfomega^{\operatorname{o}})$, where $\bfomega^{\operatorname{e}}$ (respectively, $\bfomega^{\operatorname{o}})$ is the sequence of values in even (respectively, odd) positions in  $\bfomega$, i.e., for $i\in \nat$, we have: $(\bfomega^{\operatorname{e}})_i=\bfomega_{2i}\qquad$ and $\qquad(\bfomega^{\operatorname{o}})_i=\bfomega_{2i+1}$. Then, $h$ is a homeomorphism with inverse $k:=h^{-1}:\Omega^2\to \Omega$ given by $k(\bfomega,\bfomega')=\bfomega''$, where $\bfomega''_{2i}=\bfomega_i$ and           $\bfomega''_{2i+1}=\bfomega'_i$ for $i\in \nat$. (ii) For $\Omega=[0,1]$, let $h:[0,1]\to [0,1]^2$ be the Hilbert curve with the Lebesgue measure on $[0,1]$.
  
\subsection{A parametrised family of random variable monads on {\bf PER}}

We now introduce a parametrised family of monads using Cantor sets.  Let $\Sigma_n:=\{0,1,\ldots,n-1\}$ and let $X$ be a compact metric space and consider an iterated function system (IFS) consisting of $n$ contracting maps $f_i:X\to X$, for $i\in \Sigma_n$, each given a probablity weight $p_i>0$ with $\sum_{i=0}^{n-1}p_i=1$, such that $f_i[X]\cap f_j[X]=\emptyset$ for $i\neq j$. Then the unique attractor $C\subset X$  of the IFS is a Cantor set, i.e., a non-empty, perfect, compact, totally disconnected set and satisfies $C= \bigcup_{i\in \Sigma_n}f_i[C]$; in addition $C$ is the support of the unique invariant measure $\nu$ of the IFS satisfying~\cite{hutchinson1981fractals,barnsley2014fractals}:
\[\nu=\sum_{i\in\Sigma_n}p_i\nu\circ f_i^{-1}\]The map $s: \Sigma_n^\nat\to C$ given by \[s(m)=\bigcap_{n\in\nat} f_{m_0}\circ f_{m_1}\circ f_{m_2}\circ \ldots\circ f_{m_n}[I],\] gives a homeomorphism of $C$ with $\Sigma_n^\nat$ with inverse $t:=s^{-1}$. A domain-theoretic representation of generating points of $C$ uses the Plotkin power domain of the upper space of $X$~\cite{edalat1996power}. The invariant measure $\nu_C$ can also be obtained domain-theoretically~\cite{edalat1995dynamical,edalat1996power} and we have $\nu_C=\nu\circ s^{-1}$, where $\nu$ is the product measure on $\Sigma^\nat$ induced from the finite probability space $\Sigma_n$ with $p_i$ the probability of $i\in \Sigma_n$.  

Recall that in a compact topological group, the Haar measure is the unique non-trivial left-traslation-invariant measure (up to a positive constant) with respect to the group operation~\cite{alfsen1963simplified}. When $p_i=1/n$ for all $i\in \Sigma_n$, then $\nu$ is the unique Haar probability measure on $C$ with addition defined component-wise as $x+y=(x_0\oplus y_0)(x_1\oplus y_1)\ldots$ where addition mod $n$ in $\Sigma_n$ is denoted by $\oplus$. 

Observe also that $\Sigma_n^\nat$ is a compact metric space with the metric $d$ defined as $d(x,y)=0$ if $x=y$ and otherwise
$d(x,y)=1/2^m$, where $m\in \nat$ is the first natural number such that $x_m\neq y_m$. Moreover, the maps $f_i:\Sigma_n^\nat\to \Sigma_n^\nat$, for $i\in \Sigma_n$, defined by $f_i(x)=i:x$ are contracting and $\Sigma_n^\nat$ is itself the invariant attractor of the $n$ contracting maps $f_i$ for $i\in \Sigma_n$. For the Cantor set $\Sigma_n^\nat$, as in the case of $n=2$ above, put $h:\Sigma_n^\nat\to (\Sigma_n^\nat)^2$ with $h(\bfomega)=(\bfomega^{\operatorname{e}},\bfomega^{\operatorname{o}})$, which gives a monad.
We can now define, for the attractor $C$ of the IFS with $n$ maps, $h_C:C\to C^2$ by $h_C=\langle s\circ\pi_2,s\circ \pi_1\rangle \circ h\circ t$, which gives a monad for the probability space $(C,\nu_C)$. Note that for $X=[0,1]$ and $n=2$ with $f_0:x\mapsto x/3$ and $f_1:x\mapsto x/3+2/3$, we obtain the canonical Cantor set.

\subsection{Valuation monad of {\bf PER}}
  
We now introduce the continuous valuation PER domain of a PER domain. First, we need to define the PER domain of the lattice of open sets of a PER domain.

\begin{definition}
Given a PER domain $(D, \sim_D)$, the lattice  $\mathcal{O}(D)$  of its open sets is a PER domain with the partial equivalence relation defined by $O \sim_{\mathcal{O}(D)} O'$, if for any pair of equivalent elements, $d \sim d'$  we have $d \in O$ iff $d' \in O'$.

Given a PER domain $(D, \sim_D)$, define its valuation power domain $\mathcal{V}(D, \sim_D)$ as the PER domain $(\mathcal{V}(D), \sim_{\mathcal{V}(D)})$ where $\mathcal{V}(D)$ is the domain of continuous valuations on $D$, while $\sigma \sim_{\mathcal{V}(D)} \sigma'$  iff for every pair of equivalent open sets $O \sim_{\mathcal{O}(D)} O'$ we have that $\sigma (O) = \sigma' (O')$.
\end{definition}

It is immediate to check that $\sim_{\mathcal{O}(D)}$ and $\sim_{\mathcal{V}(D)}$ are closed under lub. Note that the R-topology consists of the self-related open sets in $(\mathcal{O}(D), \sim_{\mathcal{O}(D)})$.

On morphisms, given $[f] : (D, \sim_D) \to (E, \sim_E)$, we define
$ \mathcal{V}( [f] )= [\lambda \sigma , O \,.\, \sigma (f^{-1}(O))]$. It is straightforward to check that if $f \sim f'$, $\sigma \sim \sigma'$, and $O \sim O'$ then $\sigma (f^{-1}(O)) = \sigma' (f'^{-1}(O'))$ and therefore the action on morphisms is well-defined.

By repeating the standard construction of the continuous valuation monad on domains as in~\cite{kirch1993bereiche}, $\mathcal{V}$, itself a variation of Giry's monad, is a monad on PER domains with $\epsilon'_{(D,\sim)}: (D,\sim) \to \mathcal{V}(D,\sim)$ given by $\epsilon'_{(D,\sim)}=[ \lambda d . \delta(d)]$ and flattening operation $\mu'_{(D,\sim)} :\mathcal{V}(\mathcal{V}(D,\sim))\to \mathcal{V}(D,\sim)$ given by $\mu'_{(D,\sim)} = [ \lambda \kappa , O \,.\, \int_{\sigma \in (\mathcal{V}(D), \sim)} \sigma (O)\,\,d\kappa ]$. 
More explicitly by $\int_{\sigma \in (\mathcal{V}(D), \sim)} \sigma (O)\,\,d\kappa$ we denote the integration along the measure uniquely induced by the continuous valuation $\kappa$. 

It is immediate to check that if $d \sim_D d'$ then $\delta (d )\sim_{\mathcal{V}( D)} \delta (d')$ and therefore $\epsilon'$ is a well-defined equivalence class. 

On the other hand, if $O \sim O'$, $\kappa \sim \kappa'$, we observe that for any real number $a$ the subsets $W_a, W'_a$ of $\mathcal{V}(D, \sim)$ defined by
$W_a = \{ \sigma \mid \sigma(O) > a \}$, $W_a' = \{ \sigma \mid \sigma(O') > a \}$ are equivalent open sets ($W_a \sim W_a'$). it follows that $\kappa (W_a) = \kappa'(W'_a)$. 
By Choquet's formula for conversion of the Lebesgue integral to a Riemann integral~\cite{mesiar2010choquet}, we then have:
\begin{multline*}
    \mu_{(D,\sim)}'(\kappa)(O)=\int_{x\in\mathcal{V}(D,\sim)}x(O)\,\,d\kappa=\int_0^1 \kappa(W_a)\,da \\
    =\int_0^1 \kappa'(W'_a)\,da=\mu_{(D,\sim)}'(\kappa')(O')
\end{multline*}
which implies $\mu'$ is well-defined. The commutativity of the monadic diagrams for $\epsilon'$ and $\nu'$ follow directly from those of $\mathcal{V}$ in~\cite{kirch1993bereiche}. 

\begin{definition}
    The map $T_{{\mathcal R}(D)}:{\mathcal R}(D)\to \mathcal {V}(D)$ is defined as the push-forwrd $r\mapsto \nu\circ r^{-1}$. 
\remove{Two random variables $r,s\in (\Omega\to D)$ are {\em equivalent}, written $r\sim s$, if $\nu\circ r^{-1}=\nu\circ s^{-1}$.}
\end{definition}
  We often write $T$ for $T_{{\mathcal R}(D)}$ since its type is clear from the context. The map $T$ extends that in~\cite{DiGianantonio2024} from random variables defined on probability spaces to random variables defined on sample sapces equipped with arbitrary (computable) measures. Since $T$ maps step functions to simple valuations, it is a computable map with respect to effective structures on $\mathcal{R}(D)$ and $\mathcal{V}(D)$. If $\nu$ is an unbounded valuation then $T$ is a surjection. 

  If we restrict to continuous normalized (probability) valuations, we write $\mathcal{V}(D,\sim)$ as $\mathcal{P}(D,\sim)$. The monad $\mathcal{P}$ then extends the normalised probabilisitc power domain monad to ${\bf PER}$.

\subsection{Natural transformation}

Since $\mathcal{R}$ and $\mathcal{V}$ are monads on different categories---namely, {\bf PER}, i.e., PER bounded complete domains and {\bf PER}-$\D$, i.e., PER continuous domains---it is not possible to define a natural transformation between them. However, using the forgetful functor $\mathcal{U}$ from {\bf PER}
 to {\bf PER}-$\D$ we can take the two compositions $\mathcal{U} \circ \mathcal{R}$ and $\mathcal{V} \circ \mathcal{U}$, which are both functors from {\bf PER} to {\bf PER}-$\D$ and note that $\phi : \mathcal{U} \circ \mathcal{R} \to \mathcal{V} \circ \mathcal{U}$ defined by  $\phi_{(D,\sim)} = T_{\mathcal{U}(\mathcal{R}(D,\sim))\to \mathcal{V}(\mathcal{U}(D,\sim))} = \lambda r , O \,.\, \nu( r^{-1} O)$ is a natural transformation as it can be easily checked. 
 
 The map {$T_{\mathcal{R}^2(D,\sim)\to \mathcal{V}^2(D,\sim)}$} acts on simple random variables to give simple valuations in $\mathcal{V}(\mathcal{V}(D,\sim))$:
\begin{equation}\label{double-random}
   \sup_{i\in I}(\sup_{j\in J_i}d_{ij}\chi_{V_{ij}})\chi_{U_i}\longmapsto \sum_{i\in I}\nu(U_i)\delta\left(\sum_{j\in J
    _i}\nu(V_{ij})\delta(d_{ij})\right),\end{equation}
where we have assumed the basic open sets $U_i$ and $V_{ij}$ are disjoint for $i\in I$ and $j\in J_i$.

\begin{lemma}\label{double-random-d} Let $D_0:=\operatorname{Im}(r)$, where $r=\sup_{i\in I}\chi_{U_i}(\sup_{j\in J_i}d_{ij}\chi_{V_{ij}})\in \mathcal {R}^2(D,\sim)$, then
    \[T_{\mathcal{R}(D,\sim)\to \mathcal{V}(D,\sim)}\circ \mu_{(D,\sim)}:\]\[r\mapsto\sum_{d\in D_0}\sum_{d_{ij}=d}\nu(V_{ij})\nu(U_i)\delta(d):\mathcal{R}^2 (D,\sim)\to \mathcal{V}(D,\sim)\]
\end{lemma}
\begin{proof}
   Note $\mu_D(r)=\sup_{i\in I}\chi_{h_1^{-1}(U_i)}(\sup_{j\in J_i}d_{ij}\chi_{h_2^{-1}(V_{ij})})$. Since $h_1$ and $h_2$ are independent and measure preserving, the result follows. 
\end{proof}
 
 We can now show that $\phi$ preserves the monadic structures of the monads $\mathcal{R}$ and $\mathcal{P}$ in the following sense.
\begin{proposition}
   Consider $\epsilon$ and $\mu$, respectively $\epsilon'$ and $\mu'$, the unit and flattening operations of $\mathcal{R}$, respectively $\mathcal{P}$. Then, for $(D,\sim)\in \text{PER}$, the following two diagrams commute:
\begin{center}

\begin{tikzcd}[row sep=huge, column sep=huge]
\mathcal{U} (D,\sim)\arrow[rd, "\epsilon'_{\mathcal{U}( D)}" ] \arrow[r, "\mathcal{U} (\epsilon_D)"] & \mathcal{U}(\mathcal{R}(D,\sim))
    \arrow[d,"\phi_{D}"]& \\
    &\mathcal{V}(\mathcal{U}( D,\sim))
\end{tikzcd}

\bigskip

\begin{tikzcd}[row sep=huge, column sep=large]
\mathcal{U} (\mathcal{R}^2 (D,\sim))\arrow{r}{\mathcal{U} (\mu_{(D,\sim)})} \arrow[swap]{d}{\mathcal{V}(\phi_{(D,\sim)}\circ \phi_{\mathcal{R}( D,\sim)})} 
&  \mathcal{U} (\mathcal{R}(D,\sim)) \arrow{d}{\phi_{(D,\sim)}} \\
 \mathcal{V}^2(\mathcal{U}( D,\sim))\arrow{r}{\mu'_{\mathcal{U }(D,\sim)}}&\mathcal{V}(\mathcal{U}( D,\sim))
\end{tikzcd} 
\end{center}
 
\end{proposition}

\begin{proof}
 The commutativity of the first diagram is straightforward. To show that the second diagram commutes, it is sufficient, by Scott continuity, to prove it for a simple random variable  
 \[r=\sup_{i\in I}(\sup_{j\in J_i}d_{ij}\chi_{V_{ij}})\chi_{U_i}\in \mathcal{R}^2(D,\sim)\]
 By Lemma~\ref{double-random-d}, we have 
 \[\phi_{(D,\sim)} \circ \mathcal{U} (\mu_{(D,\sim)}(r))=\sum_{d\in D_0}\sum_{d_{ij}=d}\nu(U_i)\nu(V_{ij})\delta(d)\]
 By Equation~(\ref{double-random}), we let:
 \begin{multline*}
\kappa:= (\mathcal{V}( \phi_{(D,\sim)})\circ \phi_{\mathcal{R}(D,\sim)}) (r) \\
=\sum_{i\in I}\nu(U_i)\delta\left(\sum_{j\in J
    _i}\nu(V_{ij})\delta(d_{ij})\right)
 \end{multline*}

    If $O\in \mathcal{O}(D,\sim)$, then we have 

\begin{multline*}
        \mu_{(D,\sim)}'(\kappa)(O)=\int_{x\in \mathcal{V}(D,\sim)}x(O)\,\,d\kappa \\
        =\sum_{i\in I,j\in J_i}\sum_{d_{ij}\in O}\nu(V_{ij})\nu(U_i) 
\end{multline*}
    
It follows, as required, that\\ $\mu_{(D,\sim)}'(\kappa)=\sum_{d\in D_0}\sum_{d_{ij}=d}\nu(V_{ij})\nu(U_i)\delta(d)$.
\end{proof}

\section{Conditional probability }
 Using the sigma algebra of the unit interval as events, Ackerman et al.~\cite{ackerman2019} have shown that two computable random variables on the unit interval can have a conditional distribution that encodes the halting problem, and hence non-computable. In this section, we introduce a domain-theoretic framework for conditionaly probability, based on a notion of open sets as observable property,  and show that it is computable.

 \subsection{Probability of an event}

 Following the principles of observational or semi-decidable logic~\cite{smyth1983power,abramsky1991domain,vickers1989topology}, we take only the Scott open sets of a domain $D\in \BC$, rather than all its Borel subsets, as obervable properties or events. This is the approach to probabilistic computation used in~\cite{DiGianantonio2024},
 
 However, if we naively take $\mathcal{O}(D)$ as our domain of events, we cannot develop a computable notion of conditional probability as we now see. Suppose $\sigma\in \mathcal{P}(D,=)$ is a continuous probability valuation and 
 $O,U\in \mathcal{O}(D)$ are two events. Then, we have the classical conditional probability $\operatorname{Pr}(U|O)=\operatorname{Pr}(U\cap O)/\operatorname{Pr}(O)=\sigma(U\cap O)/\sigma(O)$ if  $\operatorname{Pr}(O)=\sigma(O)\neq 0$ and $0$ otherwise. Then, $\operatorname{Pr}(\cdot\,|O):\mathcal{O}(D)\to ([0,1],\leq)$ defines a continuous valuation. However, $\operatorname{Pr}(U|O)$ is not monotonic in its second argument, which implies we cannot obtain a Scott continuous map for the conditional property. 

 The problem can be resolved if we include in our formulation, not just the obervable event $O$ but also its opposite event. In classical measure theory the opposite event is simply the complement of the event, but here the complement of $O$ is not open and thus we need to take an open set $O'$ disjoint from $o$ as is opposite. we are then led to the domain of disjoint pairs of opens $(O,O')$ of $D$, ordered by the subset inclusion of open sets. The domain of disjoint pairs of open subsets of a topological space was called the solid domain of the topological space and first used for developing a computable model of objects in solid modelling and computational geometry~\cite{edalat2002foundation}. 
 
 This domain gives a proper framework for observable events, including an interval notion of probability. 
In this paper, however, we take a simpler approach and retain the notion of probability as a real number rather than a real interval. We can then rewrite the conditional probability, for $\sigma(O)\neq 0$, as follow:
\[\operatorname{Pr}(U|(O,O'))=\sigma(U\cap O)/(1-\sigma(O'))\]
which is now monotonically incraasing in $(O,O')$ as required providing a Scott continuous and computable map in the second component; see Theorem~\ref{prob-comp}.

\subsection{Domain of disjoint events}
 Since our main purpose here is focused on probability theory, we will use the phrase {\em event domain} rather than the original name, the solid domain.
\begin{definition}
For a PER domain $(D,\sim)$, its {\em event domain}, denoted $\mathcal{E}(D,\sim)$, is the dcpo of pairs of elements $(O_1, O_2)\in (\mathcal{O}(D))^2$, with $O_1\cap O_2=\emptyset$ called {\em disjoint event-pairs}, partially ordered componentwise by the lattice order with $(O_1,O_2)\sim (O_1',O_2')$ if $O_1\sim O_1'$ and $O_2\sim O_2'$.

\end{definition} 
\begin{proposition}
  \begin{itemize}
  \item[(i)] For a PER domain $(D,\sim)$, the event-domain $\mathcal{E}(D,\sim)$ is a PER domain. 
     \item[(ii)] We have: $(O_1,O_2)\in \mathcal{E}(D,\sim)$ is maximal iff 
     $O_1=(O_2^c)^\circ$ and $O_2=(O_1^c)^\circ$.
    
   \end{itemize} 
\end{proposition}

We define the {\em conditional random variable operator}: \[\label{crv}\operatorname{CR}:\mathcal{R}( D,\sim)\times \mathcal{E}(D,\sim)\to \mathcal{R}( D,\sim)\]by 
 \begin{multline}
 \label{rvc}\operatorname{CR}(r,(O_1,O_2))(\bfomega)
 \\
= \left\{\begin{array}{cl}
  r\circ  h_1\circ h_2^n(\bfomega)    & \exists n\in \nat (i<n \\
  & \Rightarrow r \circ h_1\circ h_2^i(\bfomega)\in O_2\;\\&\& \; r \circ h_1\circ h_2^n(\bfomega)\in O_1)\\
     \bot & \mbox{otherwise}
 \end{array}\right.    
 \end{multline}
 
 It is easily seen that $\operatorname{CR}(r,(O_1,O_2)$ and $\operatorname{CR}$ are continuous. It is also clear that  $\operatorname{CR}(r,(O_1, O_2))$ is a step function if $r$ is a step function. We write $r_{(-|(O_1,O_2))}:=\operatorname{CR}(r,(O_1, O_2))$.

 An informal characterization of $\operatorname{CR}$ is the following: from the sample value $\bfomega$ one derives the infinite sequence of sample values $\langle h_1 (\bfomega), h_1\circ  h_2 (\bfomega), h_1\circ  h_2^2 (\bfomega), \ldots \rangle $ and sequentially searches the first value $\bfomega'$ in the sequence with $r(\bfomega') \in O_1$.  Note that the map $\bfomega \mapsto \langle h_1 h_2^i (\bfomega) \rangle_{i \in \nat} $ is a measure-preserving map from the sample space to the countable product of sample spaces with the product measure. 

We say $(O_1,O_2)\in \mathcal{E}(D)$ is a {\em classical} disjoint event-pair with respect to $\sigma\in \mathcal{P}(D,\sim)$, if $(O_1,O_2)$ is maximal and  $\sigma(O_1\cup O_2)=1$.
 \begin{proposition}\label{cond-rand}
 Suppose $\bot\notin U(\mathcal{O},\sim)$ and $(O_1,O_2)\in \mathcal{E}(D,\sim)$.
 \begin{enumerate}
     \item If $\nu(O_1)>0$, then we have: \[\nu(r^{-1}_{(-|(O_1,O_2))}(U))=\nu(r^{-1}(U\times O_1))/(1-\nu(r_2^{-1}(O_2))).\]
\item For $\nu(O_1)=0$, we have $\nu(r^{-1}_{(-|(O_1,O_2))}(U))=0$.
\item If $(O_1,O_2)$ is a classical event with respect to the probability measure $\nu\circ r^{-1}$, then we have:\\ $\nu(r^{-1}_{(-|(O_1,O_2))}(U))=\nu(r^{-1}(U\times O_1))/\nu(r_2^{-1}(O_1))$.
 \end{enumerate}
 \end{proposition}
{
\begin{proof}
(i)  In Equation~(\ref{rvc}), with $O_1$ is replaced by $U\cap O_1$, the events for different $n\in \nat$ are disjoint, and 
 \begin{align*}
     \nu(r^{-1}_{(-|(O_1,O_2))}(U)) & =\sum_{n=0}^\infty \nu(r^{-1}(U\cap O_1))(\nu(r_2^{-1}(O_2)))^n
     \\
     & =\nu(r^{-1}(U\cap O_1))/(1-\nu(r_2^{-1}(O_2))).
 \end{align*}
  Items (ii) and (iii) are straightforward. 
\end{proof}}

From item Proposition~\ref{cond-rand}(iii), it follows that $\operatorname{CR}(r,(O_1,O_2))$ extends the notion of conditional random variable to general disjoint event-pairs, which coincides with the classical notion when the disjoint event-pair is classical.

 Similarly, we define the {\em conditional probability}: \begin{equation}\label{cond.prob}\operatorname{CP}: \mathcal{P}(D,\sim)\times \mathcal{E}(D,\sim)\to \mathcal{P}(D,\sim)\end{equation} by 
 \begin{multline*}
     (\operatorname{CP}(\sigma,(O_1,O_2)))(U) \\
     =\left\{\begin{array}{cl} 1& \bot \in U\\
     \sigma(O_1\cap U)/(1-\sigma(O_2)))& \mbox{otherwise}.\end{array}\right.
 \end{multline*}

\begin{proposition} \label{CRandCP}
 The map $\operatorname{CP}$ is Scott continuous, and the following diagram commutes:  
 \[ \begin{tikzcd}[row sep=huge, column sep=large]
\mathcal{R}(D,\sim)\times \mathcal{E}(D,\sim) \arrow{r}{\Large{T_{\mathcal{R}(D)}\times \operatorname{Id}}} \arrow[swap]{d}{\operatorname{CR}} & 
 \mathcal{P}(D,\sim)\times \mathcal{E}(D,\sim) \arrow{d}{\operatorname{CP}} \\%
 \mathcal{R}(D,\sim)\arrow{r}{{T_{\mathcal{R}(D,\sim)}}}& 
\mathcal{P}(D,\sim)
\end{tikzcd}
\]
\end{proposition}
\subsection{Computability of the conditional probability}
We note that an effective structure on the domain $(D,=)$ induces an effective structure on the domain $\mathcal{R}(D,=)$ by using Equation~\ref{way-below}, which shows the decidability of $f\ll g$ for step functions $f,g\in \mathcal{R}(D)$. It also induces an effective structure on $ \mathcal{V}(D,=)$ using Proposition~\ref{kirch}(i) which gives the decidability of the way-below relation on simple valuations. It furthermore provides an effective structure on $\mathcal{P}(D,=)$ based on~\cite[Proposition 3.5]{edalat1995domain}, which establishes the decidability of the way-below relation on normalised simple valuations. Finally, an effective structure on $\mathcal{O}(D,=)$ clearly induces an effective structure on $\mathcal{E}(D,=)$. An effective structure for $(D,=)$ also directly provides an effective structure for $(D,\sim)$. In general, an effective structure for $(D,=)$ does not lead to an effective structure for $\mathcal{O}(D,=)$, though it does so for $D=\realDom$ and domains constructed from it using product and function space constructors.
\begin{theorem}\label{prob-comp}
    If $(D,\sim)$ and $\mathcal{O}(D,\sim)$ have effective structures, then the map $\operatorname{CP}$ is computable. 
\end{theorem}
{
\begin{proof}
    The effective structure on $(D,\sim)$ induces an effective structure on $\mathcal{P}(D,\sim)$. We need to show that for disjoint basic disjoint open sets $(O_1,O_2)\in \mathcal{O}(D,\sim)$ and a simple normailised valuation $\sigma\in \mathcal{P}(D,\sim)$, the predicate $\sigma'\ll_{\mathcal{P}(D)} \operatorname{CP}( \sigma, (O_1,O_2))$ is decidable for simple normalised valuation $\sigma'\in \mathcal{P}(D)$. This follows from Proposition 2.6 in~\cite{DiGianantonio2024} since the values of $\operatorname{CP}( \sigma, (O_1,O_2))$ are rational numbers and there are only a finite number of comparisons to check.  
\end{proof}
}
The map $\operatorname{CR}$ is also computable, which follows from a more general result, namely Proposition~\ref{CR0}, in the next section.

\section{Extended Random variables Monad}

To give semantics to the \emph{score operator}, which is commonly used to implement Bayesian inference or, more generally, score sampling, it is necessary to extend the semantics.  In particular, the semantics of an expression is associated not to a probability distribution but instead to a general measure. In our case, this generalization is achieved by transitioning from sample spaces of total measure $1$ to sample spaces having arbitrary  measure.
These measures are obtained by means of Scott continuous functions  $w : \Omega \to \overline{\realLine^+}$.
Given a sample space $(\Omega,\nu)$, equipped with the measure $\nu$, a weight function $w : \Omega \to \overline{\realLine^+}$ induces the measure $\nu_w$ defined by $\nu_w (O) =\int_{O} w \, d \nu $, for any measurable set $O$. 
\remove{Indeed, Equation~\ref{rel-density-int} can be easily checked to hold for simple functions $w$, which by the monotone convergence theorem then gives the equality for all measurable $w$. }


\begin{definition}
Given a PER domain $(D, \sim)$, the PER domain of extended random variables,  $\mathcal{R}_0(D, \sim)$, is the domain $(\Omega \to \overline{\realLine^+})  \times(\Omega \to D)$, with the partial equivalence relation defined by $(w_1, r_1) \sim (w_2, r_2)$ iff (i) $\forall \bfomega\in \Omega.\,r(\bfomega)\sim_D r(\bfomega)\,\&\,s(\bfomega)\sim_D s(\bfomega)$, and (ii) for every R-open set $O$ in $D$, $\nu_{w_1} (r_1^{-1}(O)) = \nu_{w_2}( r_2^{-1}(O))$.    
\end{definition}

Any continuous extended random variable defines a continuous valuation.

\begin{definition}
   Define the map $T_{0 \mathcal{R}_0(D,\sim)}: \mathcal{R}_0{{(D, {\sim)}}} \to \mathcal{V}(D, \sim)$, from extended random variables on $(D, \sim)$ to (non-normalized) continuous valuations on $(D, \sim)$ by
$ T_{0 \mathcal{R}_0(D, \sim) } (w, r) (O) = \nu_w (r^{-1} (O)) $. 
\end{definition}
We often write $T_0$ for $T_{0 \mathcal{R}_0(D, \sim) }$.  
Thus, $T_0 (w,r)$ is the pushforward measure of $\nu_w$ under $r$, denoted also by $r_*\, \nu_w$. Since ${T}_0(w,r)= F(w,\nu\circ r^{-1})=F(w,T(r))$, it follows from Proposition~\ref{int-on-basis}(ii) and computability of $T$ that ${T}_0$ is computable. By definition, the mapping $T_0$ preserves the equivalence relation on extended random variables.

It is easy to check that $\mathcal{R}_0$ is a functor, with the action of morphisms defined by: given $f : D_1 \to D_2$, $\mathcal{R}_0 f (w, r) = (w, f \circ r)$.

The functor $\mathcal{R}_0$ becomes a monad by letting 
$\eta_D : (D,  \sim~) \to \mathcal{R}_0(D, \sim)$ as: 
\[ \eta_D d = ((\lambda \bfomega \,.\, 1) , \ (\lambda \bfomega \,.\, d))\] and defining
$\mu_D : \mathcal{R}_0^2(D,  \sim) \to \mathcal{R}_0(D, \sim)$ as:
\[\begin{array}{ll}
\mu_D (w, r)  = &( (\lambda \boldsymbol{\bfomega} \,.\,  w (h_1 (\bfomega)) \cdot (\pi_1 (r (h_1 (\bfomega))))(h_2 (\bfomega)) , \\
     & \ \, (\lambda \bfomega \,.\, \pi_2 (r (h_1 (\bfomega)))(h_2 (\bfomega)  \ )
\end{array}
\] 
we note that in the above definition $r$ belongs to the domain $\Omega \to (\Omega \to \overline{\realLine^+}, \Omega \to D))$.

The correctness of the multiplication $\mu$ uses the following.

\begin{proposition}
If $h : \Omega \to (\Omega \times \Omega)$ is a measure preserving function, then $h$ is also a measure preserving function from $(\Omega, {(w_1 \circ h_1)  \cdot (w_2 \circ h_2)})$ to $(\Omega, {w_1}) \times (\Omega, {w_2})$, where by the expression $(w_1 \circ h_1)  \cdot (w_2 \circ h_2)$ we denote the weight function obtained by point-wise application of the product.
\end{proposition}

Any random variable can be readily transformed into an extended random variable. In fact, given $r : \Omega \to (D, \sim)$, the extended random variable $(\lambda \bfomega . 1, r)$ induces the same distribution as $r$, i.e., $T_0(\lambda \bfomega . 1, r) = T(r)$.

Let $r: \Omega \to \realDom$ be a random variable whose image lies on the maximal points $\operatorname{Max}(\realDom)$, i.e., it defines a function $r: \Omega \to \realLine$. If $r$ has probability density function $f: \realLine \to \realLine$, and $t : \realLine \to \realLine^+$ is a continuous function, the extended random variable $(t \circ r, r)$ induces a continuous valuation $T_0(t \circ r, r)$ defined by the density function $\lambda d . f(d) \cdot t(d)$. From this fact, we build the following example.
\begin{example}
By the Box-Muller transform, a standard normal random variable \( Z_0 \) with mean $0$ and standard deviation $1$ is given by:
\[
Z_0(\bfomega) = \sqrt{-2 \ln (h_1 \bfomega)} \cdot \cos(2 \pi (h_2 \bfomega))
\]
The generalized random variable $(\lambda \bfomega . 1, Z_0)$ induces a continuous valuation defined by the density function of the normal distribution $\mathcal{N}(0, 1)$, that is:
\[T_0 (\lambda \bfomega . 1, Z_0) (a,b) = \int_a^b \frac{1}{\sqrt{2 \pi}}\exp{\left( -\frac{x^2}{2} \right)}\,dx\]

Using the normal distribution $\mathcal{N}(1, 1)$, with density
$f(x) = \frac{1}{\sqrt{2 \pi}} \exp{\left({(x - 1)^2}/{2} \right)}$, one can construct the generalized random variable given by $(f \circ Z_0, Z_0)$. It induces in turn the continuous valuation $T_0 (f \circ Z_0, Z_0)$, which
is characterized by the density function:
$ g(x) = \frac{1}{\sqrt{2 \pi}} \exp{\left( -\frac{x^2}{2} \right)} \frac{1}{\sqrt{2 \pi}} \exp{\left( -\frac{(x-1)^2}{2} \right)}$, 
simplified to: \[ g(x) = \frac{1}{2 \pi} \exp{-\left( x - \frac{1}{2} \right)^2}.\]
Note that $g$ corresponds to an unnormalized density function of the normal distribution \(\mathcal{N}\left( \frac{1}{2},  \left( \frac{1}{2}  \right)^2 \right) \), which is the posterior probability obtained from a prior normal distribution \(\mathcal{N}(0,1)\) conditioned on a normal distribution \(\mathcal{N}(1,1)\).
  
\end{example}

\medskip 

We next show that the conditional probability can be given also on extended random variables on PER domains $(D,\sim)$. Let us start by extending Equation~(\ref{rvc}).
\[\label{crv0}\operatorname{CR_0}:\mathcal{R}_0(D, \sim)\times \mathcal{E}(D,\sim)\to \mathcal{R}_0 (D,\sim)\]
$\operatorname{CR_0}((w,r),(O,O') )= (w', r')$ where 
\begin{multline*}\label{ervcn} (w'(\bfomega),r'(\bfomega))= \\
 \left\{\begin{array}{cl}
  (w\circ h_1\circ h_2^n(\bfomega), \, r \circ h_1\circ h_2^n(\bfomega))    & \exists n\in \nat \,.( i<n \Rightarrow \\
  &  r\circ  h_1\circ h_2^i(\bfomega)\in O' ,\\
  & \,\,  r \circ h_1 \circ h_2^n(\bfomega)\in O\\
     (0, \bot) & \mbox{otherwise}
 \end{array}\right.
 \end{multline*}

\begin{proposition}
  The map $\operatorname{CR_0}$ is self-related.
\end{proposition}

The map $\operatorname{CR_0}$ can be defined via the fixed-point operator in the following way: consider the Scott continuous map
\begin{multline*}
\Psi: (\mathcal{R}_0(D, \sim))\times\mathcal{E}(D,\sim)\times(\mathcal{R}_0(D,\sim)) \\
\to (\mathcal{R}_0(D,\sim))
\end{multline*}
\vspace{-5ex}
\begin{multline*}
  \Psi((w,r), (O_1,O_2), (w', r'))(\bfomega) \\
  =\left\{\begin{array}{ll}
  (w h_1 (\bfomega), r h_1(\bfomega))   &  r h_1(\bfomega)
\in O_1\\  (w' h_2(\bfomega), r' h_2(\bfomega))   &   r h_1(\bfomega)\in O_2\\
  (\lambda \bfomega \, . 0, \bot)  & \mbox{otherwise}
\end{array}\right.
\end{multline*}
We note that if $(w,r)$ and $(w',r')$ are step functions, then $\Psi((w,r), (O_1,O_2), (w', r'))$ is a step function for any basis element $(O_1,O_2)\in \mathcal{E}(D)$. Note also that $\operatorname{CR}_0((w,r),(O_1,O_2))$ is the least fixed point of $\Psi(((w,r),(O_1,O_2), -):\mathcal{R}_0(D) \to \mathcal{R}_0(D) )$ with  $(w',r')\mapsto\Psi((w,r), (O_1,O_2), (w', r'))$. 

Note also that, as in the previous section, an effective structure on $(D,=)$ induces an effective structure on the domain $\mathcal{R}_0(D,=)$ and thus on the PER domain $\mathcal{R}_0(D,\sim)$.

\begin{proposition}\label{CR0}
    The map $\operatorname{CR}_0$ is computable. 
\end{proposition}
{
\begin{proof}It is easy to check that $\Psi$, as an element in the domain $\mathcal{R}_0(D,\sim) \to \mathcal{E}(E ,\sim)\to \mathcal{R}_0(D,\sim)\to \mathcal{R}_0(D,\sim)$, is computable as a composition of the computable functions used in its definition. In fact, for any basic step functions $(w,r),(w',r'),(w'',r'')\in \mathcal{R}_0(D,\sim)$, and basis element $(O_1,O_2)\in \mathcal{E}(D,\sim)$, the relation $(w'',r'')\ll \operatorname{CR}_0((w,r),(O_1,O_2),(w',r'))$ is decidable since the latter is a computable step funcition. It follows that the map $\operatorname{CR}_0$ is computable. 
\end{proof}
}
Let $w=w'={\bf 1}$, the constant function with value $1$, to obtain the computability of $ \operatorname{CR}$ of the preceding section. 

On the domain of continuous valuations, we introduce the notion of {\em conditional valuation}: 
\begin{equation}\label{cond.eval}\operatorname{CV}_0: \mathcal{V}(D,\sim)\times \mathcal{E}(D,\sim)\to \mathcal{V}(D,\sim) \end{equation}
\begin{multline*}
    \operatorname{CV}_0(\alpha,(O_1,O_2))(U) \\
 = \left\{
 \begin{array}{ll}
   \begin{array}{ll} 
 \alpha(D)& \mbox{ if } \ \ \bot \in U\\
 0 & \mbox{ if } \ \ \bot \notin U \ \wedge \alpha(O_2) = \alpha(D) \\
 \end{array} & \\
 \alpha(O_1\cap U)/(1-(\alpha(O_2)/ \alpha(D))& \mbox{otherwise}
 \end{array} \right.  
\end{multline*}
 
  
\begin{theorem} \label{CVandCP}
 The map $\operatorname{CV}_0$ is Scott continuous and computable with respect to effective structures on $\mathcal{R}_0(D,\sim)$, $\mathcal{E}(D,\sim)$ and $ \mathcal{V}(D,\sim)$, and the following diagram commutes:   
 \[ 
 \begin{tikzcd}[row sep=huge, column sep=large]
\mathcal{R}_0(D,\sim)\times \mathcal{E}(D,\sim) \arrow{r}{\Large{T_0\times \operatorname{Id}}} \arrow[swap]{d}{\operatorname{CR_0}} 
& \mathcal{V}(D,\sim)\times \mathcal{E}(D,\sim) \arrow{d}{\operatorname{CV}_0} 
\\
\mathcal{R}_0(D,\sim)\arrow{r}{{T_{0}}} 
&  \mathcal{V}(D,\sim)
\end{tikzcd}
\]
\end{theorem}
{
\begin{proof}
 The Scott continuity and the commutativity of the diagram are straightforward. Note that for a simple valuation $\alpha$, the outpout $\operatorname{CV}_0(\alpha,(O_1,O_2))$ takes only a finite number of rational numbers. The computability of $\operatorname{CV}_0$ , for effective structures on  $\mathcal{V}(D,\sim)$ and $\mathcal{E}(D,\sim)$ follows, similar to Proposition~\ref{prob-comp}, for $\operatorname{CP}$. This time the decidability of the relation in Proposition~\ref{kirch}(i) for continuous valuations is required, which follows again since there are only a finite number of comparisons between rational numbers to check. 
\end{proof}
}
In the following, we consider pairs of dependent random variables and the conditional probability of one with respect to the values assumed by the other. For such cases it is convenient to use a derived notion of $\operatorname{CR}_0$ as follows:
\[\operatorname{CR}_0':\mathcal{R}_0(D \times E, \sim)\times \mathcal{E}(E,\sim)\to \mathcal{R}_0 (D,\sim)\]
$\operatorname{CR}_0' ( (w, (r,s)), (O, O')) = (w, r')$ \ \  if \\ $\operatorname{CR_0} ((w, (r,s)), (D \times O, D \times O')) = (w, (r', s')).$

\section{Bayesian updating}

In this section, we show that the conditional probability can be evaluated using the Bayes rules, on the hypothesis that the probability density of the likelihood is known.

In our domain setting, we first formulate the notions of density function and conditional density function.
Given a PER domain $(D, \sim)$, with a continuous valuation $\alpha$, and an extended random variable $(w, r) \in \mathcal{R}_0 (D, \sim)$, a \emph{density function}, $f_{(w, r)}$, for $(w, r)$, with respect to $\alpha$, is a continuous function in $(D, \sim) \to (\overline{\realLine^+}, =)$ such that for every R-open set $O$ in the $(D, \sim)$ one has:
\[ T_0(w, r)(O) = \int_{O}  f_{(w,r)} \,  d\alpha \] 

Note that not  every extended random variable $(w,r)$ has a corresponding density function.  In particular, no simple random variable on $\realDom$, with the standard continuous valuation, has a density function.  
Moreover, any function $f'$ that coincide with $f_{(w,r)}$ on a subset of  full measure is also a density function for $(w, r)$.

Given an extended random variable $(w,(r,s))$ on the domain $(D \times E, \sim)$ with a continuous valuation $\alpha$, a \emph{density function of the conditional probability} $f_{(w,\, r \mid s)}$, with respect to $\alpha$, is a continuous function $(D \times E) \to \overline{\realLine^+}$ for which there exists a subset $S \subseteq (D \times E)$ of full measure such that for any pair $x, y \in S$, if $f_{(w, s)}(y) \neq 0$, the following equality holds:
 
\[f_{(w,\, r | s)}(x,y) = \frac{f_{(w, (r,  s))}(x,y)}{f_{(w,s)}(y)}\]

We present an example of how classical conditional density functions are translated into our domain theory setting.
\begin{example}

Consider the case where $D = E = (\realDom, =)$, the PER domain of compact intervals with equality of real numbers as PER, and the continuous valuation $\alpha$ on $\realDom$ defined by $\alpha(O) = \lambda(O \cap \realLine)$, where, with an abuse of notation, we denote by $\realLine$ the set of maximal elements of $\realDom$, and $\lambda$ is the Lebesgue measure on the real line. By definition, the set $\realLine$ has full measure.

Given a pair of classical continuous random variables $X, Y : \Omega \to \realLine$, and a pair of continuous functions $f_X, f_{X|Y} : \realLine \to [0,\infty)$ that are, respectively, the probability density function of $X$ and the conditional probability density function of $X$ given $Y$, the pair $(\mathbf{1}, X)$ and $(\mathbf{1}, Y)$, where $\mathbf{1}: \Omega \to \overline{\realLine^+}$ is the constant function $1$, defines a pair of extended random variables on $\realDom$. 

Any continuous extension $f_{(\mathbf{1}, X)}, f_{(\mathbf{1}, X|Y)} : \realDom \to ([0,\infty),\leq)$ of $f_X, f_{X|Y}$ defines the density function of $(\mathbf{1}, X)$ and the conditional density function of $(\mathbf{1}, X)$ given $(\mathbf{1}, Y)$, respectively. The function $f_{(\mathbf{1}, X)}$ can be chosen as the maximal continuous extension of $f_X$, defined by:
\[
f_{(\mathbf{1}, X)}([a, b]) = \min \{ f_X(z) \mid z \in [a, b] \},
\]
and similarly for $f_{(\mathbf{1}, X|Y)}$.
\end{example}

\remove{-------------------------------

Alternative to the above definition (where E(D) denotes the domain of disjoint pairs of open sets of D:

 The \emph{density function of the conditional probability} $f_{(w,\, r | s)}:D_1\times \mathcal{E}(D_2)\to \overline{\realLine^+}$ is given by

 \[f_{(w,\, r | s)}(x,(O_1,O_2)) 
 = \frac{\int_{O_1}f_{(w, (r,  s))}(x, \_)\, d \alpha_E}
    {\int_\Omega w\,d\nu - \int_{O_2} f_{(w, s)} d \alpha_E}  \]

    where $\alpha_{E}(O)$ is the marginal measure define by $\alpha_{E}(O) = \alpha(D \times O)$.
-------------------------------------}

Next, we relate the conditional probability density function $f_{(w,\, r | s)}$ to the
conditional probability map $\operatorname{CR}_0'$.

\begin{lemma} \label{CR-PDF}
    Given $(w, (r, s)) \in  \mathcal{R}_0 ( D \times E, \sim)$ and a continuous valualtion $\alpha$ on $( D \times E, \sim)$, if the density functions ${f_{(w, (r,  s))}(x,y)}, {f_{(w,s)}(y)}$ exists, then:
    \begin{multline*}
    T_0 (\operatorname{CR}_0' (w,(r,s)),  (O_1, O_2))(O) = \\
    \frac{\int_{O \times O_1} f_{(w,\, (r, s))} \, d\alpha}{1 - (\int_{O_2} f_{(w, s)}\, d\alpha_{E} / \int_\Omega  w \, d \nu )}    
    \end{multline*}
\end{lemma}
{
\begin{proof}
    The following chain of equations holds:

    \begin{align*}
      &  \ T_0 (\operatorname{CR}_0' (w,(r,s)), (O_1, O_2))(O) \\
    =  & \   T_0 (\operatorname{CR}_0 (w,(r,s)), (D\times O_1, D\times O_2))(O \times E) \hfill \ \ \ \\ & \ \text{by Theorem \ref{CVandCP}}\\
    =  &  \ \operatorname{CV}_0 (T_0(w, (r,s)),   (D\times O_1, D\times O_2))(O \times E)\\
    =  & \ \frac{T_0 (w,(r,s))((O \times E) \cap (D\times O_1) )}{1 - (T_0 (w,(r,s))(D \times O_2)/ (T_0 (w,(r,s))(D \times E))}   \\
    =  &  \ \frac{T_0 (w,(r,s))(O \times O_1)}{1 - (T_0 (w,s)(O_2) / \int_\Omega  w \, d \nu )}  \\
    =  &  \  \frac{\int_{O \times O_1} f_{(w,\, (r, s))}\, d\alpha}{1 - (\int_{O_2} f_{(w, s)}\, d\alpha_E / \int_\Omega  w \, d \nu )} 
    \end{align*}
\end{proof}
}


For a real number $a\in \realLine$ and $\epsilon>0$, we use the notation $a\pm\epsilon:=[a-\epsilon,a+\epsilon]$

{
\begin{lemma}\label{int-cond-rand-var}
For any extended random variable $(w, (r, s)) \in  \mathcal{R}_0 ( D \times E, \sim)$, continuous valuation $\alpha$ on $( D \times E, \sim)$, conditional density function ${f_{(w, (r |  s))}(x,y)}$, open set $O \subseteq D$ and element $e \in E$,  such that 
$f_{(w, s)}(e) \neq 0$,
we have:
\begin{enumerate}
    \item  if $\alpha_E(\dua e) \neq 0$ and $\alpha_E(\dua e \cap ((\dua e )^c)^\circ)^c = 0$ then 
\begin{multline*}
  \int_{O} f_{(w,\, r | s)}(x,e) d \alpha_D(x) \\
  \leq (T_0 (\operatorname{CR}' (w,(r,s)),  (\dua e, ((\dua e )^c)^\circ) (O)) / \int_\Omega  w \, d \nu
\end{multline*} 
\item 
if

\begin{itemize}
\item
the conditional density function $f_{(w, r | s)}$ is bounded on the subset $O \times \{e\}$,
\item
there exists a subset $M$ of full measure in $(D \times E)$, such that the restriction of the conditional density function on $M$, i.e., $f_{(w, r | s)}|_M : M \to {\realLine^+}$, is continuous with respect to the Euclidean topology on ${\realLine^+}$,
\item
there exists a base $B$ for $E$ such that $\forall b \in B \,.\, \alpha_E(\dua b \cap ((\dua b )^c)^\circ)^c = 0$ 
\end{itemize}
then, for any $\epsilon > 0$ there exists $b \in B$ such that: $e \in \dua b$

\begin{multline*}
T_0 (\operatorname{CR}' (w,(r,s)),  (\dua b, ((\dua b )^c)^\circ) (O) / \int_\Omega  w \, d \nu  \\
  \in \left( \int_{O} f_{(w,\, r | s)}(x,e) d \alpha_D(x) \right)\pm \epsilon
\end{multline*}     
\end{enumerate}

\end{lemma}
\begin{proof}
Point (i) is proved by the following chain of relations:
\begin{align*}
  & T_0 (\operatorname{CR}_0' (w,(r,s)),  (\dua b, ((\dua b )^c)^\circ)(O)) / \int_\Omega  w \, d \nu
\\
= &\frac{\int_{O \times \dua b} f_{(w,\, (r, s))}(x,y) \, d\alpha_D(x) d\alpha_E(y)}
    {\int_\Omega  w \, d \nu - \int_{((\dua b )^c)^\circ} f_{(w, s)}(y)\, d\alpha_E(y) } 
\\
= &\frac{\int_{O} (\int_{\dua b} f_{(w,\, (r \mid s))}(x,y) \cdot f_{(w, s)}(y)\, d\alpha_E(y)) \, d\alpha_D(x)}
    {\int_{\dua b} f_{(w, s)}(y) \, d\alpha_E(y)}  
    \\
\geq &\frac{\int_{O} (\int_{\dua b} f_{(w,\, (r \mid s))}(x,b) \cdot f_{(w, s)}(y)\, d\alpha_E(y)) \,  d\alpha_D(x)}
    {\int_{\dua b} f_{(w, s)}(y) \, d\alpha_E(y)} 
    \\
= &\int_{O}  f_{(w,\, (r \mid s))}(x,b) \cdot  \frac{\int_{\dua b}f_{(w, s)}(y) \, d\alpha_E(y)}
    {\int_{\dua b} f_{(w, s)}(y) \, d\alpha_E(y)} \, d\alpha_D(x)
    \\
= &\int_{O}  f_{(w,\, (r \mid s))}(x,b) \,  d \alpha_D(x) \\
\end{align*}

 Point (ii), let $\epsilon' = {\epsilon}/{(3 \cdot \alpha_D(O))}$, by continuity, for any $d \in O$ there exists a neighbourhood, $O_i$, of $d$,  and a base approximation $b_i$ of $e$, such that $\forall (d', e') \in (O_i \times \dua b_i) \cap S \,.\, | f_{(w,\, (r \mid s))}(d',e) - f_{(w,\, (r \mid s))}(d',e')| \leq \epsilon'$. Let $k > 1$ be a bound for the set of real values $f_{(w,\, (r \mid s))}[(O, e)]$, and let $I'$ be a finite subset of indexes such that, having defined $O' = \bigcup_{i\in I'} O_i$ one has $\alpha_D(O) - \alpha_D(O') < {\epsilon}/{(3 \cdot k)}<\epsilon/3$, let $b = \bigsqcup_{i \in I'}b_i$. By repeating the first steps in the previous chain of equations, we have: 

\begin{align*}
  & T_0 (\operatorname{CR}_0' (w,(r,s)),  (\dua b, ((\dua b )^c)^\circ)(O)) / \int_\Omega  w \, d \nu
\\
= &\frac{\int_{O} (\int_{\dua b} f_{(w,\, (r \mid s))}(x,y) \cdot f_{(w, s)}(y)\, d\alpha_E(y)) \, d\alpha_D(x)}
    {\int_{\dua b} f_{(w, s)}(y) \, d\alpha_E(y)}  
    \\
\in &\frac{\int_{O'} (\int_{\dua b} f_{(w,\, (r \mid s))}(x,e) \cdot f_{(w, s)}(y)\, d\alpha_E(y)) \,  d\alpha_D(x)}
    {\int_{\dua b} f_{(w, s)}(y) \, d\alpha_E(y)}  \pm 2 \epsilon /3
    \\
= &\int_{O'}  f_{(w,\, (r \mid s))}(x,e) \,  d \alpha_D(x) \pm 2 \epsilon /3 
\\
\subseteq  &\int_{O}  f_{(w,\, (r \mid s))}(x,e) \,  d \alpha_D(x) \pm \epsilon \\
\end{align*}

\end{proof} 
}

{
In the case where $E$ is the domain of real numbers $\realDom$, or a finite product of $\realDom$, and the measure $\alpha_E$ is the measure inherited from the Lebesgue measure on $\mathbb{R}$, the conditions in point (ii) are satisfied: the set of maximal elements has full measure and $\alpha_E(\dua b \cap ((\dua b )^c)^\circ)^c = 0$ is satisfied by any element $b$ in $E$. It follows that the condition also holds if the measure $\alpha_E$ is absolutely continuous with respect to the measure on the finite products of $\mathbb{R}$ inherited from the Lebesgue measure.
}

{
\begin{lemma}\label{relative-density}
For any extended random variable $(w, r) : \Omega \to (D, \sim)$ and a continous function $f (D, \sim ) \to \overline{\realLine^+ }$, the following equality among continuous valuations holds:
\[
 \, F (T_0(w,r)) \, f = T_0( (w \cdot (f \circ r)), \ r)
\]
\end{lemma}
\begin{proof}
Lemma~\ref{relative-density} follows by letting $g=f\circ r$ in Equation~\ref{rel-density-int}. In fact: 
$$
\begin{aligned}
& F \, (T_0(w,r)) \, f \, O = F \, (r^*\nu_{w}) \, f \, O \ = \ \int_O f \, d (r^*\nu_{w}) \\
& = \int_{r^{-1}(O)} f \circ  r \, d \nu_w \ = \ \int_{r^{-1}(O)}(f \circ r) \cdot w \, d\nu  \\
& = T_0( (w \cdot (f \circ r)), \ r) O
\end{aligned}
$$
\end{proof}
}
There are two different approaches to evaluating conditional probability: a general approach that uses the $\operatorname{CR}'_0$ operator, discharging the traces of computation where the condition is not satisfied, and a second, more efficient approach that applies Bayes' rule. The latter assumes that the dual conditional density function is known and uses the score operator.
The following theorem states that these two approaches yield consistent results.
\begin{theorem} \label{bayes-rule}
Given an extended random variable $(w, (r, s)) \in  \mathcal{R}_0 (D \times E, \sim)$  such  that the dual random variable  $(w, (s, r)) \in  \mathcal{R}_0 (E \times D, \sim)$, has a conditional density function  $f_{(w, s | r)}$,  we have that, for any  value $e \in E$, such that $f_{(w,s)(e) \neq 0}$ we have that 
\begin{enumerate}
 \item  if $\alpha_E(\dua e) \neq 0$ and $\alpha_E(\dua e \cap ((\dua e )^c)^\circ)^c = 0$ then 
\begin{multline*}
  T_0 ((\lambda \bfomega \, .\,  w(\bfomega) \cdot f_{(w, s | r)}(b, r(\bfomega) ) / f_{(w,s)} (b)), r) \\
  \sqsubseteq T_0 (\operatorname{CR}_0' (w,(r,s)),  (\dua e, ((\dua e )^c)^\circ)  / \int_\Omega  w \, d \nu
\end{multline*} 
\item 
if 
\begin{itemize}
\item
the conditional density function $f_{(w, r | s)}$ is bounded on the subset $O \times {e}$,
\item
there exists a subset $M$ of full measure in $(D \times E)$, such that the restriction of the conditional density function to $M$, i.e., $f_{(w, r | s)}|_M : M \to \overline{\realLine^+}$, is continuous with respect to the Euclidean topology on $\overline{\realLine^+}$,
\item
there exists a base $B$ for $E$ such that $\forall b \in B \,.\, \alpha_E(\dua b \cap ((\dua b )^c)^\circ)^c = 0$ 
\end{itemize}
then, for any $\epsilon > 0$ there exists $b \in B$ such that: $e \in \dua b$ and     
\begin{multline*}
  T_0 (\operatorname{CR}_0' (w,(r,s)),  (\dua b, ((\dua b )^c)^\circ) (O)) / \int_\Omega  w \, d \nu 
  \\
 \in 
    T_0 ((\lambda \bfomega \, .\,  w(\bfomega) \cdot f_{(w, s | r)}(e, r(\bfomega) ) / f_{(w,s)} (b)), r)(O)
     \pm \epsilon
\end{multline*}     
\end{enumerate}
\end{theorem}
{
\begin{proof}
By the previous lemma and Equation~\ref{rel-density-int}, one can write: 
\[
\begin{aligned}
& = T_0( (\lambda \bfomega \, .\,  w(\bfomega) \cdot f_{(w, s | r)}(e, r(\bfomega) ) / f_{(w,s)} (e)), r) O \\
& = T_0( (w \cdot (f_{(w, s | r)}(e, \_ ) / f_{(w,s)}(e)) \circ r), r)) O \\
& = \int_O (f_{(w, s | r)}(e, x ) / f_{(w,s)}(e)) \, d (r^*\nu_{\bfomega}) (x) \\
& = \int_O (f_{(w, (s,r))}(e, x ) / (f_{(w,r)}(x) \cdot f_{(w,s)}(e))) \, d ((\alpha_D)_{f_{(w,r)}}) (x) \\
& = \int_O f_{(w, (r,s)}(x, e) / f_{(w,s)}(e) \, d \alpha_D (x) \\
& = \int_O f_{(w, r|s}(x, e) \, d \alpha_D (x) \\
\end{aligned}
\]
Both points (i) and (ii) follow from the above equality and from Lemma~\ref{int-cond-rand-var}.
\end{proof}
}
In the case where $E$ is the domain of real numbers $\realDom$, or a finite product of $\realDom$, and the measure $\alpha_E$ is the measure inherited from the Lebesgue measure on $\mathbb{R}$, the conditions in point (ii) are satisfied: the set of maximal elements has full measure and $\alpha_E(\dua b \cap ((\dua b )^c)^\circ)^c = 0$ is satisfied by any element $b$ in $E$. It follows that the condition also holds if the measure $\alpha_E$ is absolutely continuous with respect to the measure on the finite products of $\mathbb{R}$ inherited from the Lebesgue measure.

\begin{example} A standard example where the above statement can be applied is when $(w, (a,b))$ is a random variable representing the prior distribution of two coefficients in a function $g_{(a,b)}(x)$, for example the linear function $g_{a,b}(x) = ax + b$. The random variable $(w, s)$ represents the prior distribution of measurements of $g_{a,b}(1)$, with the hypothesis that the measurement is obtained using instruments affected by error $e$, whose probability distribution is described by a random variable $(w,e)$ following a normal distribution with mean $0$ and variance $\sigma^2$. The random variables $(w,(a,b))$ and $(w,e)$ are independent, while $(w,(a,b))$ and $(w,s)$ are related by the formula: $s = g_{a,b}(1) + e$.
In this case, the conditional probability is given by the formula:  
\[
f(w, s \mid (a,b)) = \frac{1}{\sqrt{2\pi\sigma^2}} \exp\left(-\frac{(s - g_{a,b}(1))^2}{2\sigma^2}\right).
\]
\end{example}

\section{A probabilistic functional language}

Next, we present a simple functional language, which we call PFL, having a probabilistic evaluation mechanism and enriched with two primitive operators for nondeterminism: $\sample$ to sample real values from a uniform distribution on the interval $[0,1]$ and a primitive $\score$ that changes the weight associated with different traces of the probabilistic computation. The language can be described as Plotkin's PCF, a simply typed $\lambda$-calculus with constants, extended with a type for real numbers and probabilistic choice.
In defining the language, we choose a minimalistic approach and introduce only those types and constructors necessary to illustrate our approach to probabilistic computation semantics.

The type of real numbers is necessary to deal with probabilistic distribution over continuous spaces.  PFL performs a simple form of exact real number computation.
In this way we can define a precise correspondence, i.e., adequacy, between the denotational semantics base on continuous domain, and the operational semantics without bothering about the errors generated by the floating-point computation, but also without introducing unrealizable assumption in the language as having a language constant for each real-number or a decidable order relation $<$ on real-numbers.
In defining the language, we choose a minimalistic approach and introduce only those types and constructors necessary to illustrate our approach to probabilistic computation semantics.


The types of the language are defined by the grammar: 
\[\tau \ ::=  \iota \ \mid \ \ o \  \mid \ \nu \ \mid \  \rho \ \mid\  \ \tau \rightarrow \tau \]
where $\iota$ is the unit type, $o$ is the type of booleans, $\nu$ is the type of natural numbers, $\rho$ is the type of real numbers.

Some languages, e.g.|~\cite{GoubaultJT23}, introduce a type constructor $D$ for probabilistic computation: in this setting, the terms of type $\rho$ can perform only deterministic computations and return single values, while terms of type $D \rho$ can also perform probabilistic computation and return a distribution.  For simplicity, we consider all terms as potentially probabilistic and therefore denoting distributions.  By standard techniques, it would be straightforward, to transform our type system into a more refined one. 
The set of expressions in the language is defined by the grammar: 
\begin{equation} \label{grammar}
   e ::=  \lo{c} \ \mid\ x^\tau \ \mid\ e_1 e_2 \ \mid\ \lambda x^\tau \ldot e \ \mid\  \lo{ite} e e_1 e_2 \ \mid \pr e \
\end{equation}
where $x^\tau$ ranges over a set of typed variables, and $\lo{c}$ over a set of constants. The constructor $\lo{ite}$ is an if-then-else operator.  The constructor $\pr :  \rho \rightarrow \rho$ represents a function projecting a value on the unit interval $[-1,1]$. The behaviour of $\pr$ is defined by $\pr(x) = \rmax(-1, \rmin(x , 1))$.  The constructor $\pr$ plays an essential role in allowing the termination of recursive definition of functions on the real type $\rho$.  

The main constantsin probabilistic computation, are:
 
 \begin{itemize}
 \item For any pair of rational numbers $a \leq b$, a real constant $\inter{[a,b]} : \rho$ representing the result of a partial computation on real numbers, returning the interval $[a,b]$ as an approximation of the exact result.

\item a sample function $\sample : \rho$ implementing a probabilistic computation returning a real number uniformly distributed on interval $[0,1]$  

\item for any function type $\sigma \to \tau$ a fixed point operator $Y_{\sigma \to \tau} : ({\sigma \to \tau} \to \sigma \to \tau) \to \sigma \to \tau$

\item  an integration functional $\integral_\rho : (\pi \rightarrow \rho) \rightarrow \rho$, giving the integral of a function on the interval $[0,1]$, 

\end{itemize}

We do not fully specify the set remaining constants. They should be chosen in such a way to guarantee the  definability of any computable function on natural and real numbers. Typically the set of constants contains the test for greater that zero functions on natural and real numbers  $(0<)^{\rho} : \rho \to o$ and some basic arithmetic operations on natural numbers and real numbers such as $\cplus^{\nu}: \nu \to \nu \to \nu$, $\cplus^{\rho} \rho \to \rho \to \rho$.

\hypertarget{operational-semantics}{%
\subsection{Operational semantics}\label{operational-semantics}}

We define a small-step operational semantics. In the operational semantics, we need to address two main problems: modeling probabilistic computation and implementing exact real number computation.

The modeling of probabilistic computation is approached using sample-based operational semantics \cite{Park2008,Lago2019}. The intuitive idea is that probabilistic computation is modeled by deterministic computation with an extra parameter that introduces a source of randomness; in our case, this extra parameter is a sequence of random bits. Each bit is independent of the previous ones and has equal probability of being $0$ or $1$.

The $\score$ operator adds an extra level of complexity, since different execution traces, generated by different sequences of random bits, may have different weights. This is handled by explicitly maintaining the weight as an additional parameter during evaluation. The probability of obtaining a value $v$ as the result of the computation is then computed by integrating the weights over all random bit sequences that lead to that result $v$.

In contrast to other approaches to probabilistic computation, we directly address the issue of real number computation by defining a way to implement exact computation on real numbers.  Exact real number computation is considered also in \cite{GoubaultJT23}, but there the operational semantics is simplified by assuming the existence of a constant for each real number. To avoid extra complexity, we use a simplistic approach. The valuation of an expression $e$ receives an extra parameter, a natural number $n$, representing the level of precision, i.e., the effort with which the computation of $e$ is performed. If $e$ is an expression of type real, the result of the computation will be a rational interval $[a,b]$. Greater values of $n$ will lead to smaller intervals and, therefore, more precise results. In a practical implementation of exact real number computation, to avoid unnecessary computation, it is necessary to choose the precision used for the intermediate result carefully. Here we neglect the efficiency problem and content ourselves with the fact that sufficiently great values of $n$ will return the final result with arbitrary precision.  This approach to the operational semantics is somewhat similar to and inspired by~\cite{ShermanMC19,bauer08}.

We implement the three points outlined above using two techniques. First, within the operational semantics rules, we extend the syntax of PFL by adding expressions of the form $\langle e \mid s, n \rangle$.
 Here $e$ is the PFL term to evaluate, and $s$ and $n$ are two parameters used in the evaluation:
\begin{itemize}
\item $s$ is the source of randomness, a string of digits generated by the random bit generator.
\item $n$ 
(or the pair $(m,n)$) 
is the level of accuracy for real value calculations when evaluating the expression $e$.
\end{itemize}
Formally the extended language is obtained by adding to the grammar defined in Equation \ref{grammar} the new production 
\[ e::=  \langle e \mid s, n \rangle   \]
with $ s \in \{0,1\}^\star$ and $n \in \nat$.

Second, the operational semantics reduction steps are of the form $(w,e) \to (w',e')$, where $e$ is an extended PFL expression and $w$ is a rational number representing the estimated weight of the current computational trace.

The operational semantics makes it possible to derive judgments of the form $(\langle e \mid s, n \rangle, w) \to^* (\inter{[a,b]}, w')$, whose intended meaning is that the expression $e$, given the source of randomness $s$ and accuracy level $n$, reduces to the interval $[a,b]$. During this reduction, the weight of the computation trace changes from $w$ to $w'$.

Higher values of the parameter $n$ imply more effort in the computation so it will always be the case that if  $n_1 \leq n_2$, and $s_1 \sqsubseteq s_2$,  $\langle e \mid s_1, n_1 \rangle \to \inter{[a,b]}$ and $\langle e \mid s_2, n_2 \rangle \to \inter{[c,d]} $ then $[a,b]  \supseteq {[c,d]}$.  We prove that the valuation of $\langle e \mid s_0,  0 \rangle, \langle e \mid s_1, 1 \rangle, \langle e\mid s_2,   2\rangle$, \ldots,  with $s_0 \sqsubseteq s_1 \sqsubseteq s_2$ \ldots,  produces a sequence of intervals each one contained in the previous one and converging to the denotational semantics of $e$.

Given a standard term $e$, the reduction steps start from an extended term in the form $(\langle e \mid s, n \rangle, 1)$. The reduction steps are of two kinds. Some reduction steps substitute some subterm $e'$ of $e$ with some extended terms of the form $\langle e' \mid s', n' \rangle$; essentially, these steps define the parameters used by the computation on $e'$. Other reduction steps perform more standard computation steps, for example, $\beta$-reduction.

 The valuation mechanism in PFL is call-by-value. A call-by-name mechanism of valuation will have a problem when a term representing a real number is given as an argument to a function, the formal argument of the function can appear in several instances inside the body of the function, and there is no way to force these several instances of the formal argument to evaluate to the same real number.
This behaviour is unnatural since it introduces in the computation more non-determinism than one would like to have.   

The computation halts on pairs $(e, w)$, where $e$ is a value. The values, $v$, on the types $o$ and $\nu$, are the Boolean values and natural number constants. The values for type $\rho$ are the rational intervals, $\inter{[a,b]}$. The values for function types are the $\lambda$-expressions and the functions constant in the language. 

The valuation contexts are standard ones for a call-by-value reduction:
\[ 
E[\_] ::= [\_] \ \mid\ E[\_] \, e  \ \mid\ v \, E[\_] \ \mid\ \lo{op} v \, E[\_]  
\ \mid\ \lo{ite} E[\_]  \, e_1 \, e_2 
\] 

For any reduction context $E[\ ]$, there is the rule:
\[
\frac{(e_0, w_0) \to (e_1, w_1)}{(E[e_0], w_0) \to (E[e_1], w_1)} 
\]
The reduction rules for application are: 
\begin{align*}
(\langle e_1 e_2 \mid s, n \rangle, w) & \to  (\langle e_1 \mid s_e, n \rangle \langle e_2 \mid s_o, n \rangle , w) 
\\ 
(\langle \lambda x^\tau \ldot e_1 \mid s, n \rangle, w) & \to  (\lambda x^\tau \ldot \langle  e_1 \mid s, n \rangle, w)
\\
((\lambda x^\tau . e) v, \, w) & \to (e [v / x^\tau],\, w)
\end{align*}
where $s_e$ and $s_o$ denote the sequences of bits appearing in the even, respectively odd, positions in $s$.

The constant $\sample$ is not a value and the associated reduction rule is:
\begin{multline*}
(\langle \sample \mid \langle s_i \rangle_{1 \leq i \leq m} , n \rangle, w) \\
\to ([\Sigma_{i=1}^{o} s_i \cdot 2^{-i}, \ (\Sigma_{i=1}^{o} s_i \cdot 2^{-i}) + 2^{-n} ] , w)  
\end{multline*}
with $o = \min(m, n)$

The rule for $\score$ is the only one explicitly modifying the weight associated to a computation:
\[
(\score [a,b], w) \to ( (), a \cdot w)
\]

The constructor $\pr$, as $\lo{ite}$, does not always force the valuation of its argument:
\begin{align*}
(\langle \pr \mid s, 0 \rangle e, \, w) & \to ([-1,1], \, w) \\
(\langle \pr \mid s, (n+1) \rangle [a,b], w) & \to
\begin{cases}
    ([-1,-1], w) & \mbox{if } b \leq -1 \\
    ([1,1], w) & \mbox{if } a \geq 1 \\
    ([a,b] \cap [-1,1], w) & \mbox{otherwise}
\end{cases}
\end{align*}

For the other constructors and constants $c$ we have the rules: 
$$
( \langle c \mid s, n \rangle, w) \to (c, w),
$$
together with specific rules concerning the constants in the language that denotes functions. 

For the fixed point operator, $Y_\sigma$, we distinguish whether the type $\sigma$ is in the form 
$\sigma = \sigma_1 \rightarrow \ldots \sigma_n \rightarrow \rho$, or not; that is, if the final return value is a real number or not.  
In the first case, the operational semantics makes use of the parameter $n$, and it is defined as
\begin{align*}
(Y_{\sigma} (\lambda F \ldot \langle e \mid s, n \rangle),\, w) & \to 
  (\langle e [(Y_{\sigma} \, \lambda F \ldot e)/ F] \mid s, (n - 1) \rangle,\, w)
\end{align*}
Intuitive explanation: in order to have a convergent computation, in an expression of the form $Y_{\sigma} (\lambda F \ldot e)$ any recursive call of $F$ should be guarded by an application of the projection operator $\pr$, that is any instance of $F$ in $e$, should be inside an expression in the form $\pr e'$. In this way, when the parameter $n$ in the valuation of the recursive call
$\langle e [(Y_{\sigma} \, \lambda F \ldot e)/ F] \mid s, n \rangle$ reaches the value $0$, the recursive call 
$(Y_{\sigma} \, \lambda F \ldot e)$ is avoided and $\pr e'$ generates the value $[-1,1]$.

In the second case we use the more standard rule: 
\[
(Y_{\sigma} (\lambda F \ldot \langle e \mid s, n \rangle),\, w) \to 
  (\langle e [(Y_{\sigma} \, \lambda F \ldot e)/ F] \mid s, n \rangle,\, w) 
\]
The operational semantics rules for the constants representing the operation on real numbers $(0 <), +, -, *, /, \rmin, \rmax$ are the standard ones, for example:
\begin{align*}
\\
  ((0 <)\inter{[a,b]},\, w) & \to (\ttt,\, w) \ \ \text{ if } a > 0
\\
  ((0 <)\inter{[a,b]},\, w) & \to (\ff,\, w) \ \ \, \text{ if }  b < 0
\\
 (\inter{[a_1, b_1]} + \inter{[a_2, b_2]},\, w) & \to (\inter{[a_1 + a_2, b_1 + b_2]},\, w)
\\
\ldots  
\end{align*}
The remaining rules are standard for a call-by-value typed $\lambda$-calculus. 


\subsection{Denotational semantics}\label{denotational-semantics}

A call-by-value $\lambda$-calculus is usually modelled considering a category of domains without bottom~\cite{Win93}, in which to interpret values, and raise objects, with bottom, in which to interpret general expressions. In our setting it is quite straightforward to define a category of PER domains without bottom $\PD_s$; for lack of space we omit the explicit definition.
The semantics interpretation is given in terms of the computational monads.  In our case, the monad of computations, transforming the domain of values, on a given type, to the domain of computation, is defined as $\mathcal{C} = \mathcal{R}_0 \circ (\ )_\bot$ where $(\ )_\bot$ is the lifting monad, necessary to give semantics to recursion, and $\mathcal{R}_0$ is the monad for random variables. The lifting monad $(\ )_\bot$ is generated by an adjunction pair from the category of PER domains without bottom $\PD_s$ to the category of PER domains with bottom, $\PD$, and strict functions as morphism, while $\mathcal{R}_0$ is a monad on the category of $\PD_\bot$, it follows that $\mathcal{C}$ can be defined as the composition of an adjunction pair and therefore is a monad. Since both $\mathcal{R}_0$ and $(\ )_\bot$ are strong commutative monads, $\mathcal{C}$ is a commutative monad too.

The natural transformations 
$\eta^\mathcal{C}_{\langle D, \sim_D\rangle} : \langle D, \sim_D\rangle \to \mathcal{R}_0(\langle D, \sim_D\rangle)_\bot$ and 
$\mu^\mathcal{C}_{\langle D, \sim_D\rangle} :   \mathcal{R}_0(\mathcal{R}_0(\langle D, \sim_D\rangle)_\bot)_\bot \to \mathcal{R}_0(\langle D, \sim_D\rangle)_\bot$ making $\mathcal{C}$ a monad are:  
\begin{align*}
   & \eta^\mathcal{C}_{\langle D, \sim_D\rangle} (d) =  (\mathbf{1}, \lambda \bfomega \,.\,\lifted{d}) \\
   & \mu^\mathcal{C}_{\langle D, \sim_D\rangle} (w,r) =  
    \\
   & \hspace{1em} \begin{aligned}[t]
           & (\lambda \bfomega \,.\, \mathrm{let}\ (w_1, r_1) \Leftarrow r (h_1 (\bfomega)) \ \mathrm{in} \ w(h_1 (\bfomega)) \cdot w_1(h_2 (\bfomega)), \\
     & \ \lambda \bfomega \,.\, \mathrm{let}\ (w_1, r_1) \Leftarrow r(h_1 (\bfomega)) \ \mathrm{in} \ r_1(h_2 (\bfomega))) 
    \end{aligned}
\end{align*}
where $\lifted{d}$ denotes the lifted element of $d$, while \\ $\mathrm{let}\, x \Leftarrow d_1 \, \mathrm{in} \, d_2(x)$, denotes the value $\bot$ if $d_1$ is $\bot$, and the value $d_2(d)$ if $d_1 = \lifted{d}$.

The  PER domain $\Dom_\tau$, used to give semantic interpretation to values having type $\tau$, is recursively defined as follows: $\Dom_o = \langle \{ \lo{ff}, \lo{tt} \}, \id_o \rangle$, $\Dom_\nu = \langle \nat, \id_\nu \rangle$, $\Dom_\rho = \langle \realDom, \id_\rho \rangle$,  $\Dom_{\sigma \to \tau} = (\Dom_\sigma \to C \, \Dom_\tau)$, where $\id_o, \id_\nu, \id_\rho$ are the identity relations on the corresponding domains.

The semantic interpretation of \emph{deterministic} constants is derived by the standard deterministic semantics of this operator by the composition of the natural transformation $\eta$.  For example the denotation of the successor operator $\lo{succ}$ on natural numbers, $\bsem{\lo{succ}}$ is an element on $\Dom_{\nu \to \nu} = \Dom_\nu \to \mathcal{C} \, \Dom_\nu$ defined by  $\eta_{\Dom_\nu} \circ \bsem{\lo{succ}}^{\mathbf{st}}$, where $\bsem{\lo{succ}}^{\mathbf{st}}$ is the standard, deterministic semantics of $\lo{succ}$.

The semantic interpretation of $\sample$, that it is not considered a value, is parametric on the sample space $\Omega$ and on a function $u : \Omega \to \realDom$ which is measure preserving:
for $\Omega = [0,1]$, $u$ is the immersion function, for $\Omega = 2^\nat$,
$u (\langle x_i \rangle_{i \in \nat}) = {\Sigma_{i=1}^{\infty} x_{i} \cdot 2^{-i}}$

\[ \esem   
{\sample}_{\sigma} : \mathcal{C} \Dom_\rho = ( \mathbf{1}, \,\lambda \bfomega \,. \lifted{u(\bfomega))}
\] 
The semantics of $\score$ satisfies the condition of being a continuous function from $\realDom \to (\Omega \to \overline{\realLine^+}$), and it is consistent with the weight function providing a lower bound for the valuation function associated with a random variable. 
\[ \esem   
{\score}_{\sigma}([a,b] ) = (\lambda \bfomega \,.\, a, \, () )
\] 

The semantic interpretation of the fixed-point operator is: 
\begin{align*}
&  \bsem{\lo{Y}_{\tau} } F  = \bigsqcup_{i \in \nat} (\mu^\mathcal{C}_{\Dom_\tau} \circ F)^i (\bot_\tau) 
\end{align*}

with $F : \Dom_{\tau \to \tau} = (\Dom_\tau \to \mathcal{C} \Dom_\tau)$

The semantic interpretation function \(\mathcal E\) is defined, by structural induction, in the standard way. 
More specifically, given a (well typed) term $e$ of type $\tau$ with free variables on the list $x_1^{\tau 1}, \ldots x_n^{\tau n}$, the denotational semantics of $e$, $\esem{e}$ is a continuous function from PER domain $\Dom_{\tau 1} \times \ldots \times \Dom_{\tau n}$ into $C \Dom_{\tau}$ preserving the partial equivalence relation, i.e. $[\esem{e}]$ is a morphism on the category of PER domains. 

\[\begin{aligned}
  \esem{c}_\rho &= \eta (\bsem{c})\\
  \esem{x^\tau}_\rho &= \eta (\rho(x^tau))\\
  \esem{e_1 e_2} &= \mu \circ (T \lo{eval}) \circ t \circ  \langle \esem{e_1}, \esem{e_2}\rangle \\
  \esem{\lambda x^\tau.e}_\rho &= \lambda d \in \mathcal D_\rho.\esem{e}_{(\rho[d/x])}\\
\end{aligned}\]

\subsection{Adequacy}\label{adequacy}

The correspondence between the denotational and the operational semantics is shown by the following result. 
Let $S$ be a finite set of disjoint binary sequences, we define the measure of $S$ as 
$\mu(S) = \Sigma_{s \in S} 2^{- \mathrm{lengh}(s)}$. 

For a rational number $q$, constant $c$ and closed FPL expression  $e$ of ground type $\tau$, i.e. $\tau \in \{ \iota, o, \nu, \rho \}$  we write \( (c,q) \ll_o Eval(e)\) if there exists a natural number $n$, a finite set of disjoint binary sequences $S$ and an $S$-indexed set of rational numbers $\{ q_s | s \in S\}$ such that 
for any $s \in S$, $(\langle e, n, s \rangle,\, 1) \to^* (d,\, q_s) $ with  $q < \Sigma_{s \in S}\, q_s$ and $c = d$ for the ground types $ \iota, o, \nu$ , for the ground type $\rho$  $d$ has to be an interval contained in the interior of the interval $c$   ($c= [a,b], d =  [a',b'] $ and $a < a' < b' < b$).

Similarly, for any PER domain $(D, \sim)$, element $d \in (D, \sim)$, $(w, r) \in \mathcal{C}(D, \sim)$, and rational number $q$, we write 
$(d, q) \ll_d (w, r)$ if $q < \nu_w r^{-1}(\dua d)$.
\begin{theorem} \label{adequacyT}
The operational semantics is sound and complete with respect to the denotational semantics, that is for any  ground type $\tau$,  constant$c$,  closed expression $e$ of type $\tau$  and dyadic number $q$, $(c, q) \ll_o Eval(e)$ iff 
$(\bsem{c}^{\mathbf{st}}, q) \ll_d \esem{e}$ 
\end{theorem}
{
\begin{proof}
We use the standard proof technique of computability predicates.

We define a computability predicate $\Comp$, first on closed terms of groung type \( \tau\), by requiring that the denotational and operational semantics coincide. That is, 
$\Comp(e)$ when for all $q, c$ $(c, q) \ll_o Eval(e)$ iff 
$(\bsem{c}^{\mathbf{st}}, q) \ll_d \esem{e}$ 

The computability predicate is then extended to closed expressions of arrow type by requiring the preservation, to closed elements of any type, and by closure, to arbitrary elements.

Using the standard technique of computability predicates, it is possible to prove that all constants are computable, and that \(\lambda\)-abstraction preserves the computability of expressions. Therefore all expressions are computable.
\end{proof}
}

\subsection{Conditional probability }
 The conditional probability can be expressed in our language in the following way.
Given a term $e_r : \tau$ describing an extended random variable $r$ on $(\Dom_\tau)_\bot$, that is, $r \in \mathcal{C}(\Dom_\tau) = (\Omega \to \overline{\realLine^+}, \Omega \to (\Dom_\tau)_\bot)$ such that $\esem{e_r} = r$, and
a deterministic term $e_t : \tau \to o$ describing a disjoint event-pair $(O_1, O_2)$, that is $\esem{e_t} \in C(\Dom_\tau \to C(\Dom_o)) = \eta^{\mathcal{C}}_{\Dom_\tau \to C(\Dom_o)}(f)$ with $f : \Dom_\tau \to C(\Dom_o)$ and $O_1 = f^{-1}(\eta^{\mathcal{C}}_{\Dom_o}(\ttt))$ and $O_2 = f^{-1}(\eta^{\mathcal{C}}_{\Dom_o}(\ff))$,
the conditional random variable $r_{(-,(O_1,O_2))}$ is equal to the denotational semantics of the term:
\[Y_{\tau} (\lambda x^{\tau} \,.\, ( \lambda y^{\tau} \,.\, \operatorname{if} \, e_t \, y^{\tau} \, \operatorname{then} \, y^{\tau} \, \operatorname{else} \, x^{\tau}) \, e_r )\]
A general function whose arguments are two variables $x^{\iota \to \tau}_r$ and $x^{\tau \to o}_t$, with $x^{\iota \to \tau}_r$ the descriptions of a random variable on $r$ $\Dom_\tau$ and $x^{\tau \to o}_t$ the description of a disjoint event-pair in $\Dom_\tau$ returns the conditional probability can be define as follows:  
\begin{multline*}
    \lambda x^{\iota \to \tau}_r \,,\, x^{\tau \to o}_t \,.\, \\
    Y_{\tau} (\lambda x^{\tau} \,.\, ( \lambda y^{\tau} \,.\, \operatorname{if} x^{\tau \to o}_t y^{\tau} \operatorname{then} y^{\tau} \operatorname{else} x^{\tau}) (x^{\iota \to \tau}_r *) )
\end{multline*}
 It is convenient to describe random variable $r$ in the domain $\Omega \to D$ using thunks, that is terms in the form $\lambda x^\iota \, . \, e$, a thunk is a to pass, in call-by-value language an expression $e$ to function $f$ without forcing its valuation, or in other terms, to simulate a call-by-name passage of arguments in a call-by-value language.  

\section{Conclusion}
We have presented a domain-theoretic framework for probabilistic programming that addresses four main issues: the treatment of conditional probability, continuous distributions,  exact computation, and higher-order functional programming. 

In particular, we have shown that in the domain-theoretic framework of open sets as observable events, the conditional probaility and conditional random variables are computable on effectively given domains.

Our framework provides a constructive definition of conditional probability through observable events and proves the consistency between scoring-based and rejection sampling approaches.

Future work includes developing interval notions of computability to effectively normalize general valuations into probabilistic valuations. Additionally, we plan to apply our semantic framework to develop formal reasoning principles for establishing equivalence between probabilistic programs.


\bibliography{CCCRandomVariables,library_ord}
\bibliographystyle{ACM-Reference-Format}

\end{document}